\documentclass[13pt,a4paper]{article}

\usepackage[ruled]{algorithm2e}
\renewcommand{\algorithmcfname}{ALGORITHM}

\usepackage{algpseudocode}
\usepackage{amsthm} 
\usepackage{amsmath}
\usepackage{amssymb}
\usepackage{mathtools}
\usepackage{rotating}
\usepackage{siunitx}
\usepackage{subfigure}
\usepackage{float}
\usepackage{hyperref}

\SetAlFnt{\small}
\SetAlCapFnt{\small}
\SetAlCapNameFnt{\small}
\SetAlCapHSkip{0pt}
\IncMargin{-\parindent}

\newtheorem{theorem}{Theorem}[section]

\newtheorem{lemma}[theorem]{Lemma}
\newtheorem{corollary}[theorem]{Corollary}
\newtheorem{definition}[theorem]{Definition}
\newtheorem{notation}[theorem]{Notation}
\newtheorem{remark}[theorem]{Remark}

\newcommand{\BigO}[1]{\ensuremath{\operatorname{O}\bigl(#1\bigr)}}

\begin{document}

\title{Practical Algorithms for Finding Extremal Sets}
\author{Martin Marinov\\ Nicholas Nash\\ David Gregg}
\maketitle

\begin{abstract}
The minimal sets within a collection of sets are defined as the ones which do not have a proper subset within the collection, and the maximal sets are the ones which do not have a proper superset within the collection. Identifying extremal sets is a fundamental problem with a wide-range of applications in SAT solvers, data-mining and social network analysis. In this paper, we present two novel improvements of the high-quality extremal set identification algorithm, \textit{AMS-Lex}, described by Bayardo and Panda. The first technique uses memoization to improve the execution time of the single-threaded variant of the AMS-Lex, whilst our second improvement uses parallel programming methods. In a subset of the presented experiments our memoized algorithm executes more than $400$ times faster than the highly efficient publicly available implementation of AMS-Lex. Moreover, we show that our modified algorithm's speedup is not bounded above by a constant and that it increases as the length of the common prefixes in successive input \textit{itemsets} increases. We provide experimental results using both real-world and synthetic data sets, and show our multi-threaded variant algorithm out-performing AMS-Lex by $3$ to $6$ times. We find that on synthetic input datasets when executed using $16$ CPU cores of a $32$-core machine, our multi-threaded program executes about as fast as the state of the art parallel GPU-based program using an NVIDIA GTX 580 graphics processing unit.
\end{abstract}

\section{Introduction}
\label{sec:intro}
\subsection{Motivation}
\label{sec:intro:back+motiv}
\noindent

The problem studied in this paper is that of finding the extremal sets within a dataset (family of sets) $D$. The extremal sets of $D$ are all the sets in $D$ that are maximal or minimal with respect to the partial order induced on $D$ by the subset relation.

Finding extremal sets is a fundamental problem and has many motivating applications. For example, large-scale SAT solvers use extremal set identification as an optimization step \cite{EenBiere05}. Extremal sets are also used for performing itemset support queries in data mining \cite{Mielikainen+06}, and social network analysis \cite{BayardoPanda11}, as well as in trajectory-based query algorithms with applications in surveillance \cite{Vieira+09}. Early theoretical algorithms were motivated by problems in propositional logic \cite{Pritchard91}.

We find our inspiration for working on the problem of finding extremal sets in the domain of searching for optimal depth sorting networks. Bundala~\textit{et al.}~\cite{BundalaCCSZ14_Optimal_Depth} describe a method (Lemma 2 in Section 3.2) for reducing the search space by considering only the output minimal networks within a collection of outputs of comparator networks of the same depth. Although, Bundala et al. present a stronger search space reduction technique --- output minimal up to permutation --- the problem of finding the minimal itemsets within a dataset is used a preliminary reduction step. The reason being that the minimal up to permutation problem is GI-Hard \cite{Marinov:II:GI-Hard} and the minimal itemset problem is known to be sub-quadratic \cite{Pritchard97}; hence, one would use the output of the latter as an input to the former. The algorithm described in this paper was initially developed to find such output-minimal networks (itemsets) within a dataset and hence, our discussion and examples focus on finding the minimal itemsets. However, as with Bayardo and Panda's state of the art practical algorithm AMS-Lex~\cite{BayardoPanda11}, our approach can be used to compute minimal or maximal itemsets.

and hence it is aimed at finding the minimal itemsets and not the maximal ones --- as per Bayardo and Panda's state of the art practical algorithm AMS-Lex~\cite{BayardoPanda11} for finding extremal (minimal or maximal) sets within a dataset.


In this paper, we present two optimization techniques that we apply to the AMS-Lex algorithm to achieve a faster execution time --- the first one uses memoization and the second one parallel programming techniques. The memoization technique is aimed at speeding up the AMS-Lex algorithm for finding the extremal itemsets within datasets containing a large number of common prefixes --- such as the ones found in the sorting networks domain. The presented parallel version of AMS-Lex is aimed at utilizing more of the CPU resources that are generally available in modern day computers. Using experimental evaluation we demonstrate the speedup achieved of both of them when compared to the highly efficient implementation of the AMS-Lex\footnote{Bayardo and Panda have made their implementation of the AMS-Lex algorithm publicly available at \url{https://code.google.com/p/google-extremal-sets/}} algorithm described by Bayardo~and~Panda~\cite{BayardoPanda11}. 

Given that AMS-Lex `is easily modified to find minimal itemsets'~\cite{BayardoPanda11}, without loss of generality, in this paper we focus on finding the minimal itemsets within an input dataset. We give full explanation on how exactly AMS-Lex is to be modified to find the minimal itemsets --- rather than the maximal itemsets \cite{BayardoPanda11} --- in Section~\ref{sec:background}. Furthermore, since our optimization techniques build on top of the existing algorithm (and implementation) of AMS-Lex the presented modification of AMS-Lex can be easily transformed to find the maximal itemsets.

\subsection{Related Work}
\label{sec:intro:related_work}
\noindent

We denote by $N$ the sum of the cardinalities of all the sets in the input dataset $D$, and informally refer to it as the size of the input. Although the algorithms for computing extremal sets are almost quadratic in $N$ in the worst case, due to the nature of datasets in applications, practical algorithms can operate efficiently for very large $N$ \cite{BayardoPanda11}. In this paper we provide experimental results for $N = 7.2 \times 10^8$.

Yellin~\cite{Yellin92} described algorithms for maintaining a dynamic family of sets, under insertion, deletion, intersection and subset query operations. He presents an output sensitive algorithm for identifying extremal sets after a sequence of $n$ operations that operates in $O(mn)$ time, where $m$ is the number of maximal sets. Note that $n$ is the sum of $N$ and the number of sets in the dataset, and hence $n > N$.

Early sub-quadratic time algorithms for finding extremal sets were described by Yellin and Jutla~\cite{YellinJutla93}, operating in $O(N^2 / \log N)$ expected time, and by Pritchard~\cite{Pritchard97} who provided a matching worst-case time bound. Pritchard~\cite{Pritchard91} described the first algorithms that required sub-quadratic space, providing algorithms requiring $O(N^2 / \log N)$ space.

Sheni and Evans~\cite{Shen96} also studied algorithms for maintaining a dynamic family of sets, operating in time $O(N^2 / \log^2 N)$ and requiring $O(N^2 / \log^3 N)$ space.
We do not study this dynamic version of the extremal set problem in this paper. 

Pritchard~\cite{Pritchard97} described the first algorithm to make use of a lexicographic ordering of the input sets. Among the practical algorithms for computing extremal sets is the highly efficient implementation of the AMS-Lex algorithm described by Bayardo and Panda~\cite{BayardoPanda11}. AMS-Lex is the state of the art practical algorithm for finding extremal sets that is designed to run on commodity CPUs. In this paper we give a detailed explanation of AMS-Lex in section~\ref{sec:background} as it is the basis point of our work.

Fort \textit{et al.}~\cite{Fort+13} described a highly parallel algorithm designed specifically for graphics processing units (GPUs). Fort~et~al. sihow that their parallel algorithm running on a GPU can outperform AMS-Lex running on single core of a conventional CPU. The single-threaded algorithm we described in this paper is targeted at running on an ordinary commodity CPU and therefore we compare its performance to the algorithm of Bayardo and Panda~\cite{BayardoPanda11}. In the experimental evaluation section~\ref{sec:experiments} we compare the execution time of our two new algorithms to Fort~\textit{et al.}~\cite{Fort+13}'s reported execution time by evaluating on synthetically generated datasets.

\subsection{Contributions}
\label{sec:intro:contributions}
The main contributions of this work can be summarized as:

\begin{itemize}
	\item A memoized version of AMS-Lex that takes advantage of common prefixes among itemsets.
    \item We outline a parallel modification of the AMS-Lex extremal sets algorithm.
    \item We present experimental results over both real-world and synthetic data for both the memoized and parallel modifications of the AMS-Lex extremal sets algorithms. We find that the speedup of the memoized algorithm increases as the length of the common prefixes of itemsets in the input dataset increases. Also that, the speedup of the parallel algorithm increases as the number of CPU cores used increase. 
\end{itemize}

\section{Background}
\label{sec:background}

Practical algorithms for computing the extremal sets of a dataset $D$ assume that the elements of $D$ are \textit{sets of items}, called \textit{itemsets}. Furthermore, these algorithms assume that there is an ordering on the \textit{itemsets} themselves. An input to an extremal set algorithm is then an ordered multiset of itemsets, referred to as a \textit{dataset} $D$.

The choice of the ordering on the itemsets gives rise to alternative algorithms for computing extremal sets. For example, if itemsets are ordered by cardinality then the simple observation that if itemset $a$ is a proper subset of itemset $b$ then the cardinality of $a$ is less than the cardinality of $b$ can be used to prune the search space. This gives rise to an algorithm referred to as \textit{AMS-Card} by Bayardo~and~Panda~\cite{BayardoPanda11}.

Pritchard~\cite{Pritchard97} exploited a lexicographic ordering of itemsets to obtain more efficient algorithms for identifying extremal sets. In particular he noted the following:

\begin{theorem}
\label{th:lex}
Let $a$ and $b$ be $itemsets$ such that $a \subset b$ then either $a$ is a $proper$ $prefix$ of $b$ or $a$ is lexicographically larger then $b$.
\end{theorem}

The most efficient practical algorithm, AMS-Lex, for identifying extremal sets, described by Bayardo~and~Panda~\cite{BayardoPanda11}, makes use of this lexicographic ordering of the preceding property to substantially prune the search space. In order to present our improvements we must first describe in detail the AMS-Lex algorithm.



\subsection{The AMS-Lex algorithm}
\label{sec:background:ams-lex}

In this section we reproduce the AMS-Lex algorithm, we re-use the notation \cite{BayardoPanda11} when referring to the input ordered dataset $D$:

\begin{itemize}
    \item $D[i]$ denotes the $i^{th}$ itemset in $D$
    \item $D[i][j]$ denotes the $j^{th}$ item of itemset $D[i]$.
    \item $D[i : j]$ denotes the ordered multiset of itemsets $\lbrace D[k] ~| k = i \ldots j \rbrace$ in that order.
    \item $D[i][j : k]$ denotes the ordered multiset of items $\lbrace D[i][l] ~|~ l = j \ldots k \rbrace$.
\end{itemize}

We also re-use Bayardo and Panda's \textit{subsumed} notation: an itemset $A$ is subsumed by $B$ iff $A$ is a subset of $B$.

\begin{algorithm} [t]
\SetAlgoNoLine
\SetKwProg{func}{Function}{}{}
\SetKwFunction{containsSubsetOf}{Contains-Subset-Of}
\func{\containsSubsetOf{$D[b:e]$, $S$, $j$, $d$}}
{
	\KwIn{The ordered multiset of itemsets $D[b:e]$, an itemset $S$ and two integers $j$ and $d$ where $b \leq e$ and $1 \leq j \leq |S|$ and $1 \leq d < |D[b]|$. The parameter $j$ specifies we need only consider $S[j : |S|]$ and $d$ is the size of the common prefix shared by all $ I \in D[b:e] $ and $S$.}
	\KwOut{Returns $true$ iff there exists a proper subset of $S$ within $D[b:e]$, and $false$ otherwise.}

	\nl \If{$S[j] < D[b][d + 1]$}
	{
    	\nl $j \leftarrow NextItem(S, j, D[b][d+1])$\;
    	
    	\nl \If{ $j$ is null }
    	{
    		\nl \KwRet $false$\;
    	}
	}
	
	\nl \eIf{$S[j] = D[b][d + 1]$}
	{
		\nl $e' \leftarrow NextEndRange(D[b:e], S[j], d+1)$\;
		
		\nl \If{$\lvert S \lvert > d + 1$ and $\lvert D[b] \lvert = d+1$}
		{
			\tcc{ $D[b]$ is a proper subset of $S$. }
			\nl \KwRet $true$\;
		}
		
		\nl \If{ $j+1 \leq \lvert S \lvert$ }
		{
			\nl \If{$\containsSubsetOf(D[b : e'], S, j + 1, d + 1)$}
			{
				\nl \KwRet $true$\;
			}
		}
		
		\nl $b \leftarrow e' + 1$\;
	}
	{
		\nl $b \leftarrow NextBeginRange(D[b:e],S[j],d)$\;
		\tcc{When there is no element in $D[b:e]$ that has a value greater than or equal to $S[j]$ at index $d + 1$ then the function $NextBeginRange(D[b:e],S[j],d)$ returns $e+1$; i.e. we can safely deduce that there is no subset of $S$ within the collection $D[b:e]$.}
	}
	
	\nl \If{{$b \leq e$}}
	{
		\nl \KwRet $\containsSubsetOf(D[b : e], S, j, d)$\;
	}
	
	\nl \KwRet $false$\;
}
\caption{Pseudo code for finding if the input dataset $D = \{  D_1, D_2, \dots, D_{n} \}$ contains a proper subset of $S$. A reproduction of the function MarkSubsumed, described by Bayardo and Panda, but used for finding the minimal itemsets rather than the maximal ones, i.e. for finding the minimal itemsets within the dataset $D$ we do not mark the subsumed itemsets but rather return true if a properly subsumed itemset by $S$ exists within $D$ and false otherwise. }
\label{algo:Contains-Subset-Of}
\end{algorithm}

\begin{algorithm} [t]
\SetAlgoNoLine
\SetKwProg{func}{Function}{}{}
\SetKwFunction{amsLex}{Get-Minimal-Itemsets-Lex}
\func{\amsLex{$D$}}
{
	\KwIn{Dataset $D = \{ D_1, D_2, \dots, D_n \}$ that is ordered lexicographically and every itemset $D_i \in D$ is also ordered lexicographically. }
	\KwOut{The minimal itemsets within the dataset $D = \{  D_1, D_2, \dots, D_{n} \}$.}
	
	\nl $bool~is\_min[n] \longleftarrow \{true, true, \dots, true\}$\;
	
	\tcc{ Find itemsets subsumed by proper prefix. }
	\nl $S \longleftarrow D[1]$\;
	
	\nl \For{$i = 2$ \KwTo $n$}
	{
		\nl \If{$ |S| \leq |D[i]| \And D[i][1 : |S|] = S $}
		{
			\tcc{ $S$ is a proper prefix of $D[i]$. }
			\nl $is\_min[i] \longleftrightarrow false$\;
		}
		\Else
		{
			\nl $S \longleftarrow D[i]$\;
		}
	}
	
	\tcc{ Find itemsets subsumed by non-proper prefix. }
	\nl \For{$i = 1$ \KwTo $n - 1$}
	{
		\nl \If{$is\_min[i] \And \containsSubsetOf( D[i+1 : n], D[i], 1, 0 ) $ \tcc{see Algorithm~\ref{algo:Contains-Subset-Of}} }
		{
			\nl $is\_min[i] \longleftarrow false$\;
		}
	}
	
	\nl \KwRet $ \{ D_i \in D$ $\lvert$ $is\_min[i] = true \}$\;
}
\caption{Pseudo code for finding the minimal itemsets within the dataset $D = \{  D_1, D_2, \dots, D_{n} \}$ by using the lexicographic constraint (Theorem~\ref{th:lex}). A reproduction of the AMS-Lex algorithm described by Bayardo and Panda, but used for finding the minimal itemsets rather than the maximal ones, i.e. for finding the minimal itemsets within the dataset $D$ we do not mark the subsumed itemsets but rather mark an itemset as non-minimal if it is a superset another one.}
\label{algo:AMS-Lex}
\end{algorithm}

The pseudo code of the AMS-Lex algorithm itself is shown in Algorithm~\ref{algo:AMS-Lex}, and it applies the result from of Theorem~\ref{th:lex} directly to first identify the proper prefixes that are subsumed by lexicographically smaller itemsets, and then searching among the remaining itemsets using Contains-Subset-Of. The function Contains-Subset-Of takes as input an itemset $S$ and dataset $D$ and returns all $x \in D$ such that $x \subset S$ and $x$ is lexicographically larger than $S$. Contains-Subset-Of makes use of the common prefixes of itemsets in $D$ as well as the lexicographic order of $D$. Since the items in the itemsets themselves are ordered lexicographically, the functions NextBeginRange, NextEndRange, and NextItem can be implemented using binary search.

\begin{figure} [t]
	\centering
	\subfigure[Call graph of Contains-Subset-Of for the itemset $D_1=abc$ over the dataset $D$]{\includegraphics	[scale=0.13]{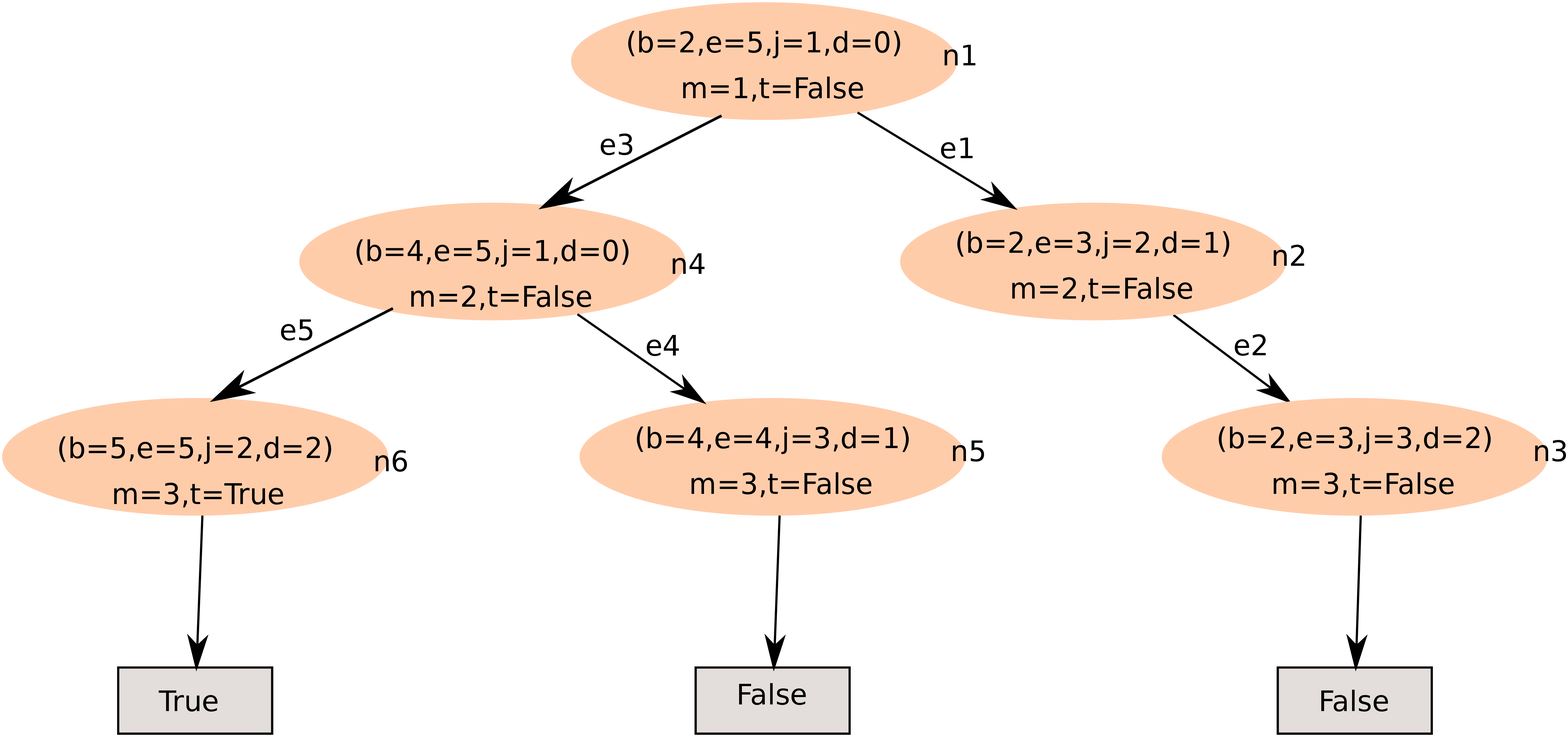}}
	\subfigure[Call graph of Contains-Subset-Of for the itemset $D_2=abde$ over the dataset $D$]{\includegraphics	[scale=0.13]{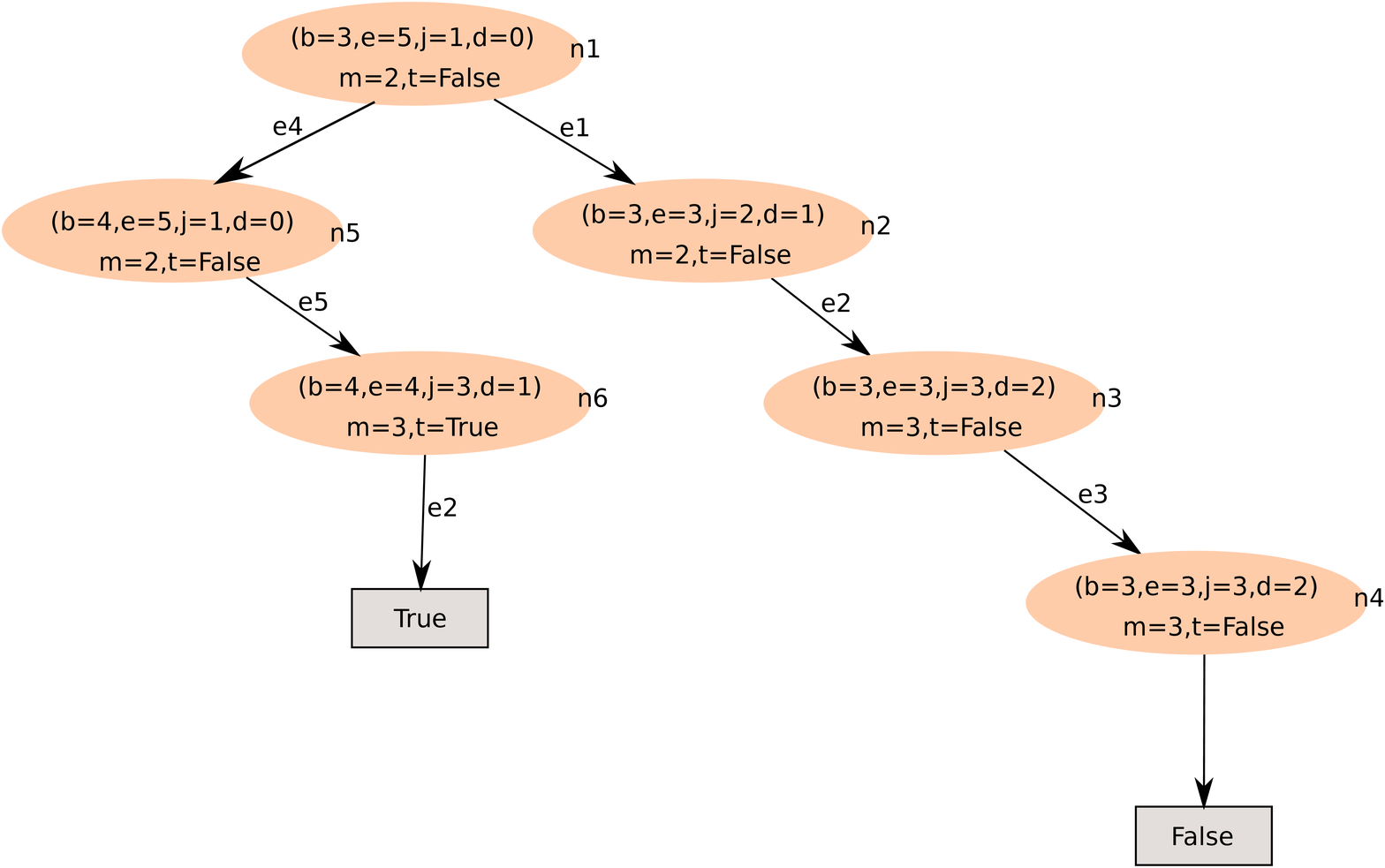}}
	\subfigure[Call graph of Contains-Subset-Of for the itemset $D_3=abdf$ over the dataset $D$]{\includegraphics	[scale=0.13]{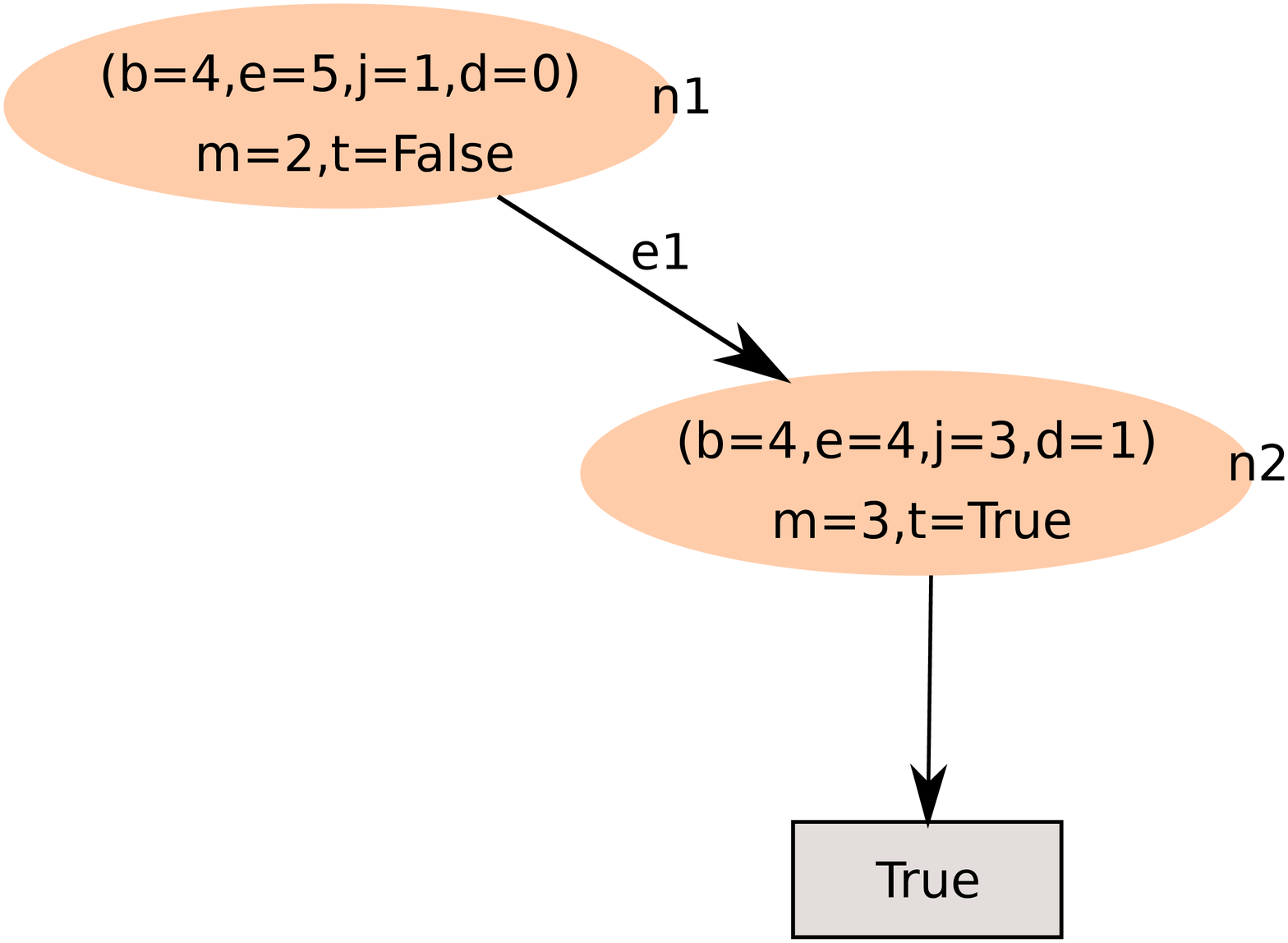}}
	\subfigure[Call graph of Contains-Subset-Of for the itemset $D_4=bd$ over the dataset $D$]{\includegraphics		[scale=0.13]{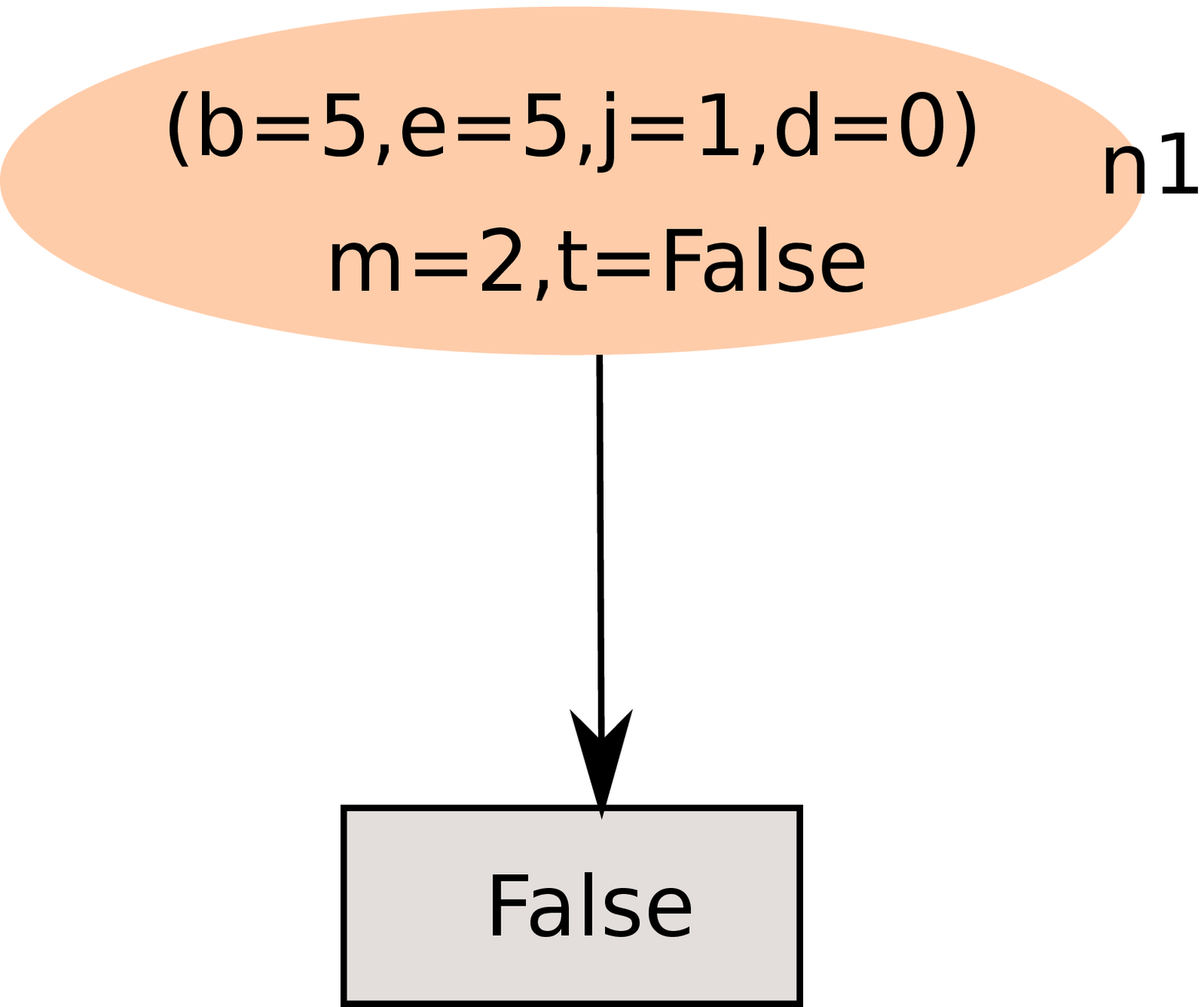}}
	\caption{ The call graphs of the AMS-Lex~\cite{BayardoPanda11} algorithm for the function Contains-Subset-Of over the dataset $D = \{ D_1 = abc, D_2 = abde, D_3 = abdf, D_4 = bd, D_5 = c \}$. All graph nodes($n_i$) and edges($e_j$) are labelled in the order of executions --- first is $n_1$, then $n_2$, then $n_3$, etc. The AMS-Lex algorithm is presented in Algorithm~\ref{algo:AMS-Lex} and Contains-Subset-Of in Algorithm~\ref{algo:Contains-Subset-Of}. Further explanation of these call graphs can be found in Section~\ref{sec:background:ams-lex:example}.}
	\label{fig:algo:ams-lex:example}
\end{figure}

\subsubsection{Example}
\label{sec:background:ams-lex:example}

Figure~\ref{fig:algo:ams-lex:example} presents the call graphs (as per Definition~\ref{def:call_graph}) of the AMS-Lex~\cite{BayardoPanda11} algorithm for finding the minimal itemsets over the dataset $D = \{ D_1 = abc, D_2 = abde, D_3 = abdf, D_4 = bd, D_5 = c \}$. Looking at the figure we can see that the minimal itemsets are $D_4$ and $D_5$; also that $D_1 \supset D_5$, $D_2 \supset D_4$ and $D_3 \supset D_4$. The dataset $D$ is chosen such that every line of the function Contains-Subset-Of is executed at least once thus handling all cases of Bayardo and Panda's~\cite{BayardoPanda11} AMS-Lex algorithm. 

\subsubsection{Contains-Subset-Of Explanation}
The Contains-Subset-Of function exploits the common prefixes of itemsets in $D$ by taking advantage of the lexicographic order of $D$. 
The function is designed to efficiently find all itemsets in the range $D[b:e]$ that are subsets of $S$ (i.e., that are subsumed by $S$).
The itemsets in $D$ are processed in ranges which share a common prefix of length at least $d$.

The first thing we check in the function is if the next item ($D[b][d+1]$) is contained in $S$ by finding the first element of S which is greater than or equal to $D[b][d+1]$. 
If all elements of $S$ smaller then $D[b][d+1]$ we can safely deduce that there are no subsumed itemsets by $S$ in the range $D[b:e]$. 
This is because all itemsets in $D[b:e]$ are ordered lexicographically in ascending order. Hence if $S[\lvert S\lvert] < D[b][d+1]$ then 
$S[\lvert S\lvert] < D[i][d+1]$ for all i in the range $[b,e]$. Hence we reach a state where we know that the element $S[j] \geq D[b][d+1]$.

If $S[j]=D[b][d+1]$ then we know that it is possible for $D[b]$ to be a subset of $S$. 
Hence we have to make a recursive call to Contains-Subset-Of. In order to do this we have to first find a new end range $e'$ such that all elements in $D[b:e']$ have a common prefix of length at least $d+1$. Then  
check if there are any subsumed itemsets. Next we check if the requirements of the recursive call to Contains-Subset-Of that we want to make are met. If this is the case then we mark 
subsumed items by $S$ in the range $D[b:e]$. Since we have already covered the range $D[b:e']$ we set the current start of our range $b$ to $e'$.

If $S[j]>D[b][d+1]$ then we know that $D[b]$ cannot be subsumed by $S$. Hence we search for the first element in $D[b:e]$ which has a value at index $d+1$ greater then or equal to $S[j]$, this operation is referred to as 
subroutine NextBeginRange.

Lastly we check if the current begin range is smaller then the current end range and if it is the case we mark all subsumed sets of $S$ in the range $D[b:e]$ by making a recursive call to Contains-Subset-Of.

\section{A Memoized Algorithm for Identifying Extremal Sets}
\label{sec:memoized}

The AMS-Lex \cite{BayardoPanda11} algorithm uses a frequency based item ordering to reduce the probability of itemsets sharing long common prefixes. Nonetheless, AMS-Lex takes advantage of common prefix shared between consecutive itemsets. More precisely, in the definition of the function MarkSubsumed~\cite{BayardoPanda11} and its variant presented here --- Contains-Subset-Of (Algorithm~\ref{algo:Contains-Subset-Of}); they both have the arguments $D[b:e]$, $S$, $j$, $d$ with the restriction that all itemsets $I \in D[b:e]$ must share a common prefix of size $d$. Hence, even after the item-based frequency ordering of the input dataset, common prefixes are expected to occur, otherwise, this logic would not have been included in AMS-Lex by Bayardo and Panda. Therefore, the current state of the art practical algorithm AMS-Lex exploits and takes advantage of the common prefixes between itemsets.

The observations and memoization technique that we present in this section are all based on the common prefixes shared by itemsets --- we take it a step further than Bayardo and Panda by analysing the behaviour of successive calls to the function Contains-Subset-Of (MarkSubsumed) by two itemsets $S'$ and $S''$ which share a non-empty common prefix; whereas the current approach \cite{BayardoPanda11} focuses on the efficient implementation of the function Contains-Subset-Of.

\subsection{Observations}
Our improved algorithm for extremal set identification memoizes successive calls to the function Contains-Subset-Of, defined in Algorithm~\ref{algo:Contains-Subset-Of}. As we explain below, Bayardo and Panda's algorithm AMS-Lex presented in Algorithm~\ref{algo:AMS-Lex} duplicates work in successive calls to Contains-Subset-Of where itemsets share a non-empty common prefix. We now show more precisely the duplicated work, in terms of the call-graphs resulting from successive calls to Contains-Subset-Of.

\begin{definition}
\label{def:call_graph}
The directed call graph of an itemset $S$ and the function Contains-Subset-Of$(D[b:e], S, j, d)$ is defined as a graph $G(S) = (V, E)$, where $V = \{(b, e, S, j, d) ~\lvert~ b, e, S, j$ and $d$ meet the input requirements of Contains-Subset-Of$\}$, and $(v_1, v_2) \in E$ iff Contains-Subset-Of$(v_1.b,~v_1.e,~v_1.S,~v_1.j,~v_1.d)$ makes a recursive call to Contains-Subset-Of$(v_2.b,~v_2.e,~v_2.S,~v_2.j,~v_2.d)$.
\end{definition}

\begin{remark}
\label{RemarkOutDegreeOfGraph}
Note that since the Contains-Subset-Of function in Algorithm~\ref{algo:Contains-Subset-Of} performs at most two recursive calls, hence the out-degree of any vertex in a call-graph $G(S)$ is at most two. 
\end{remark}

\begin{notation}
For a call graph $G(S) = (V, E)$ and any $v = (b, e, S, j, d) \in V$, we refer to the values of $v$ as $v.b$, $v.e$, $v.j$ and $v.d$; and we refer to the children of $v$ as $v.c_1$ and $v.c_2$. We denote $v.t$ as a boolean field which is true iff there exists a subset of $S$ in the range $[v.b;v.e]$ that is of size $ v.d+1 $. We denote $v.m$ as the maximum index that is accessed from the itemset $S$ without considering any recursive calls of Contains-Subset-Of.
\end{notation}

\begin{remark}
\label{MartinRemark}
Note that at any \textit{single} call-graph node corresponding to a call to function Contains-Subset-Of$(D[b:e], S, j, d)$ the only indices of $S$ that are required are those between $j$ and NextItem$(S, j, D[b][d+1])$. Hence, we can see that $v.m$ is bounded above by NextItem$(S, j, D[b][d+1])$.
\end{remark}

\begin{lemma}
\label{MartinLemma}
Let $S$ and $T$ be itemsets with a common prefix $P$. Let $G(S) = (V_S, E_S)$ and $G(T) = (V_T, E_T)$. 
Suppose that $v_1, v_2 \in V_S$, where $v_1 = (b, e, S, j, d)$, and $v_2 = (b', e', S, j', d')$ such that $j' < \vert P \vert$, and that $(v_1, v_2) \in E_S$. 
Then $(w_1, w_2) \in E_T$ where $w_1=(b, e, T, j, d)$ and  $w_2=(b', e', T, j', d')$.
\end{lemma}

\begin{proof}
Referring to Algorithm~\ref{algo:Contains-Subset-Of} note that because $S$ and $T$ have a common prefix $P$ of length greater than $j'$ all requirements of Contains-Subset-Of are met for the inputs represented by $w_1$ and $w_2$. Hence we have $w_1, w_2 \in V_T$. We now need to show that there is an edge between $w_1$ and $w_2$. Since $(v_1, v_2) \in E_S$ and from Remark~\ref{MartinRemark} the only values required of $S$ by Contains-Subset-Of are in the range $[j, j']$ and as a result of the further assumption that $j' < \vert P \vert$ it follows immediately that $(w_1, w_2) \in E_T$.
\end{proof}

\begin{remark}
Note that for any itemset $S$, the call graph $G(S) = (V, E)$ is acyclic because in all recursive calls to Contains-Subset-Of the range $[b, e]$ gets smaller, $S$ is always constant, $j$ increases and $d$ increases. 
\end{remark}

\begin{notation}
For any itemset $S$, we refer to the subgraph of $G(S) = (V, E)$ identified by $V'=\{(b, e, S, j, d) \in V ~\lvert~ j < \lvert P \lvert\}$ as $G(S) \lvert_{j < \lvert P \lvert}$.
\end{notation}

\begin{corollary}
\label{MartinCorollary}
Let $S$ and $T$ be $itemsets$ with a common prefix $P$.
Then $G(S) \lvert_{j < \lvert P \lvert} = G(T) \lvert_{j < \lvert P \lvert}$.
\end{corollary}

\begin{proof}
Use induction to apply Lemma~\ref{MartinLemma} multiple times starting from the root of $G(S)$ identified by the vertex $(b, e, S, j=1, d=0)$.
\end{proof}

\subsection{Algorithm}
\label{sec:memoized:algo}

\begin{algorithm} [t]
\SetAlgoNoLine
\SetKwProg{func}{Function}{}{}
\SetKwFunction{containsSubsetOfMemoized}{Contains-Subset-Of-Memoized}
\func{\containsSubsetOfMemoized{$D$, $i$, $p$, $v$}}
{
	\KwIn{ The dataset $ D = \{  D_1, D_2, \dots, D_{n} \} $. The parameter $i$ specifies that we trying to find if a proper subset of the itemset $D[i]$ exists within $D[i+1:n]$. The input also contains a call graph node $v \in G(S)$ for some itemset $S$ and same dataset $D$; and the integer $p$ --- the size of the longest common prefix between $S$ and $D[i]$. }
	\KwOut{ Returns $true$ iff there exists a proper subset of $D[i]$ within $D[i+1:n]$, and $false$ otherwise. }
	
	\nl \If{$v.m > p$}
	{
		\tcc{ The maximum index that was accessed from the method $\containsSubsetOf$ in the memoized iteration represented by $v$ is larger then the size of the common prefix, so we must invoke the the function to find the non-proper subsets of $D[i]$ as no more memoized results can be used. }
    	\nl $b \leftarrow max(v.b, i + 1)$\;
    	
    	\nl \If{$b \leq v.e$}
    	{
	    	\tcc{ We assume a modified version of the function $\containsSubsetOf$ which returns a pair consisting of a boolean variable and a node representing the call stack of the function. }
    		\nl $ \langle res, v \rangle \longleftarrow \containsSubsetOf(D[b:v.e], D[i], v.j, v.d) $\;
    		
    		\nl \KwRet $ res $\;
    	}

    	\nl $v \longleftarrow null$\;\label{algo:Contains-Subset-Of-Memoized:v_null}
    	
    	\nl \KwRet $ false $\;
    }
    
    \nl \If{$ v.c_1 \neq null $}
    {
	    \nl \If{$ \containsSubsetOfMemoized( D, v.c_1, i, p ) $}
	    {
		    \nl $res \longleftarrow true$\;
	    }
    }
    
    \nl \If{$ v.c_2 \neq null $}
    {
	    \nl \If{$ \containsSubsetOfMemoized( D, v.c_2, i, p ) $}
		{
			\nl $res \longleftarrow true$\;
		}
    }
    
    \tcc{ recall that $v.t$ equals true iff a subset was found in the execution of the function $\containsSubsetOf$ without considering the recursive calls; i.e. there exists a non-proper subset of $D[i]$ of size smaller then the length of the common prefix $p$ between $S$ and $D[i]$. }
    \nl \KwRet $ v.t $\;
}
\caption{ Pseudo code for finding if the dataset $D = \{  D_1, D_2, \dots, D_{n} \}$ contains a proper subset of $D[i]$ by using memoization --- the call graph node $v \in G(S)$ and the common prefix between $ D_i $ and $S$.  }
\label{algo:Contains-Subset-Of-Memoized}
\end{algorithm}

\begin{algorithm} [t]
\SetAlgoNoLine
\SetKwProg{func}{Function}{}{}
\SetKwFunction{amsLex}{Get-Minimal-Itemsets-Lex-Memoized}
\func{\amsLex{$D$}}
{
	\KwIn{Dataset $D = \{ D_1, D_2, \dots, D_n \}$ that is ordered lexicographically and every itemset $I \in D$ is also ordered lexicographically. }
	\KwOut{The minimal itemsets within the dataset $ D $.}
	
	\nl $bool~is\_min[n] \longleftarrow \{true, true, \dots, true\}$\;
	
	\tcc{ Find itemsets subsumed by proper prefix. }
	\nl $S \longleftarrow D[1]$\;
	
	\nl \For{$i = 2$ \KwTo $n$}
	{
		\nl \If{$ |S| \leq |D[i]| \And D[i][1 : |S|] = S $}
		{
			\tcc{ $S$ is a proper prefix of $D[i]$. }
			\nl $is\_min[i] \longleftrightarrow false$\;
		}
		\Else
		{
			\nl $S \longleftarrow D[i]$\;
		}
	}
	
	\tcc{ Find itemsets subsumed by non-proper prefix. }
	\nl $S \longleftarrow null$\;
	\nl $v \longleftarrow null$\;
	\nl \For{$i = 1$ \KwTo $n - 1$}
	{
		\nl \If{$is\_min[i] $}
		{
			
			\nl \If{$v = null$}
			{
				\tcc{ defined in Algorithm~\ref{algo:Contains-Subset-Of} but assuming that it returns a pair of a boolean value $res$ and the call stack represented by $v$. }
				\nl $\langle res, v \rangle \longleftarrow \containsSubsetOf( D[i+1 : n], D[i], 1, 0 )$\;
				
				\nl \If{$ res $}
				{
					
					\nl $is\_min[i] \longleftarrow false$\;
				}
			}
			\Else
			{
				\tcc{ largest common prefix of $S$ and $D[i]$. }
				\nl $p \longleftarrow max( \{ 1 \leq j \leq min(|D[i]|, |S|) ~|~D[i][1 : j] = S[1 : j] \} )$\;
				
				\tcc{ note that the function $\containsSubsetOfMemoized$ modifies the node $v$. }
				\nl \If{$ \containsSubsetOfMemoized( D, i, p, v ) $}
				{
					\nl $is\_min[i] \longleftarrow false$\;
				}
			}
			
			\nl $S \longleftarrow D[i]$\;
		}
	}
	
	\nl \KwRet $ \{ D_i \in D$ $\lvert$ $is\_min[i] = true \}$\;
}
\caption{Pseudo code for finding the minimal itemsets within the dataset $D = \{  D_1, D_2, \dots, D_{n} \}$ by using memoization and the lexicographic constraint (Theorem~\ref{th:lex}).}
\label{algo:AMS-Lex-Memoized}
\end{algorithm}

The pseudo code of our modified algorithm for identifying minimal sets is presented in Algorithm~\ref{algo:AMS-Lex-Memoized} and we now give an informal description of its behaviour. 
For each call made to Contains-Subset-Of$(D[i + 1, n], D[i], 1, 0)$ we memoize the call graph $G(D[i])$ of the execution path.
When we get to the point when we need to find if there is a subsumed itemset by $D[i+1]$ we first identify the common prefix $P$ of $D[i]$ and $D[i + 1]$.
Then we traverse $G(D[i])$ using depth first search. For each vertex $v$ we check if a recursive call is made to Contains-Subset-Of with some $j \geq \vert P \vert$.
If this is the case then we execute the function Contains-Subset-Of with input $v$; otherwise we recursively traverse the children of $v$. This is a direct result from Corollary~\ref{MartinCorollary}. In practice we note that, we need not memoize the full call graph $G(D[i])$ as we are only ever going to use nodes $w \in G(D[i])$ for which $w.j < |P|$.

\begin{remark}
It is important to note that we use a modified version of the function Contains-Subset-Of by assuming that it returns a pair of a boolean result as per the specification from Algorithm~\ref{algo:Contains-Subset-Of} and the call graph representing its execution path. We use this in the pseudo code of the memoized version of the memoized version of AMS-Lex presented in Algorithm~\ref{algo:AMS-Lex-Memoized}.
\end{remark}

\begin{figure} [t]
	\centering
	\subfigure[The memoized call graph $v$ after processing the itemset $D_1=abc$over the dataset $D$]{\includegraphics[scale=0.13]{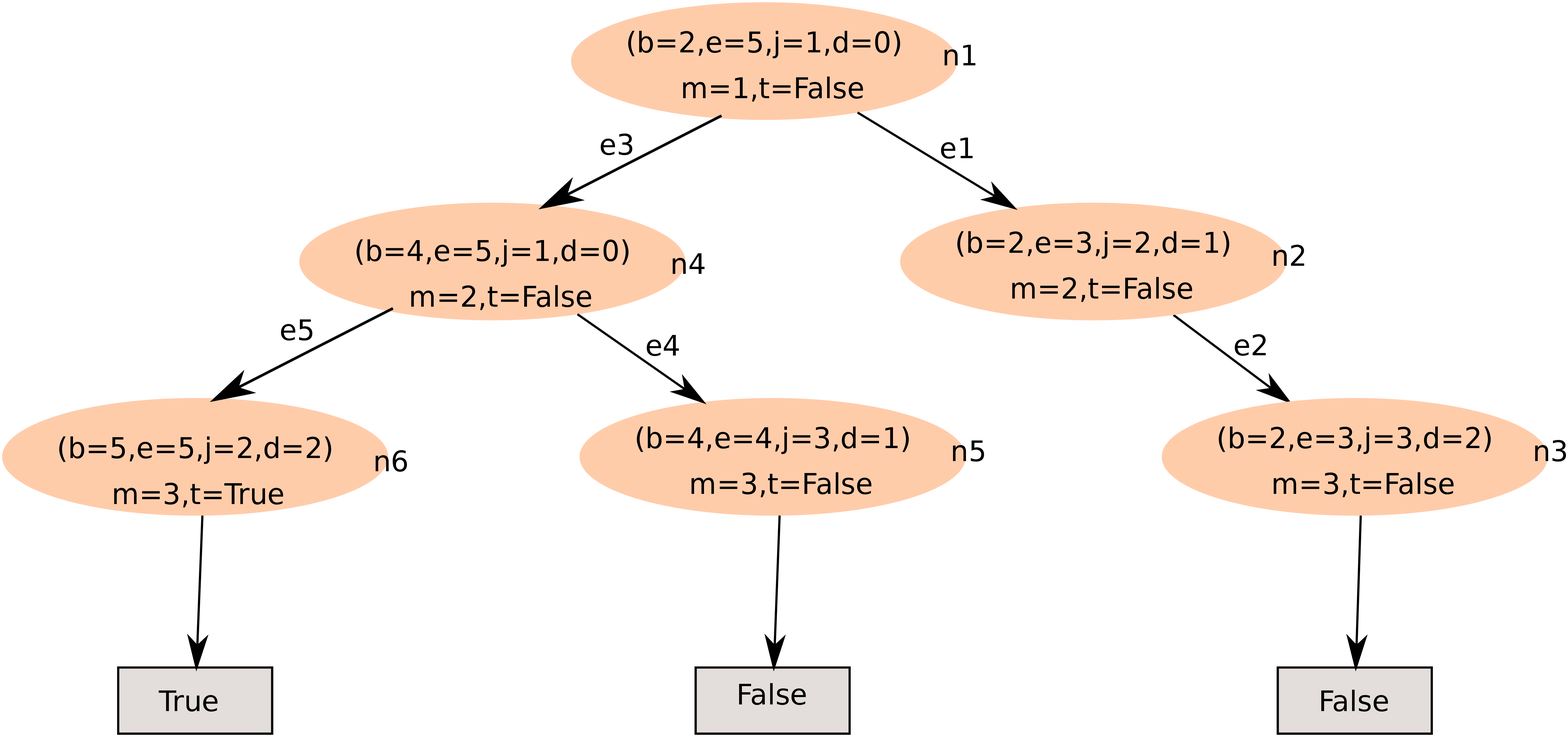}}
	\subfigure[The memoized call graph $v$ after processing the itemset $D_2=abde$ over the dataset $D$]{\includegraphics[scale=0.13]{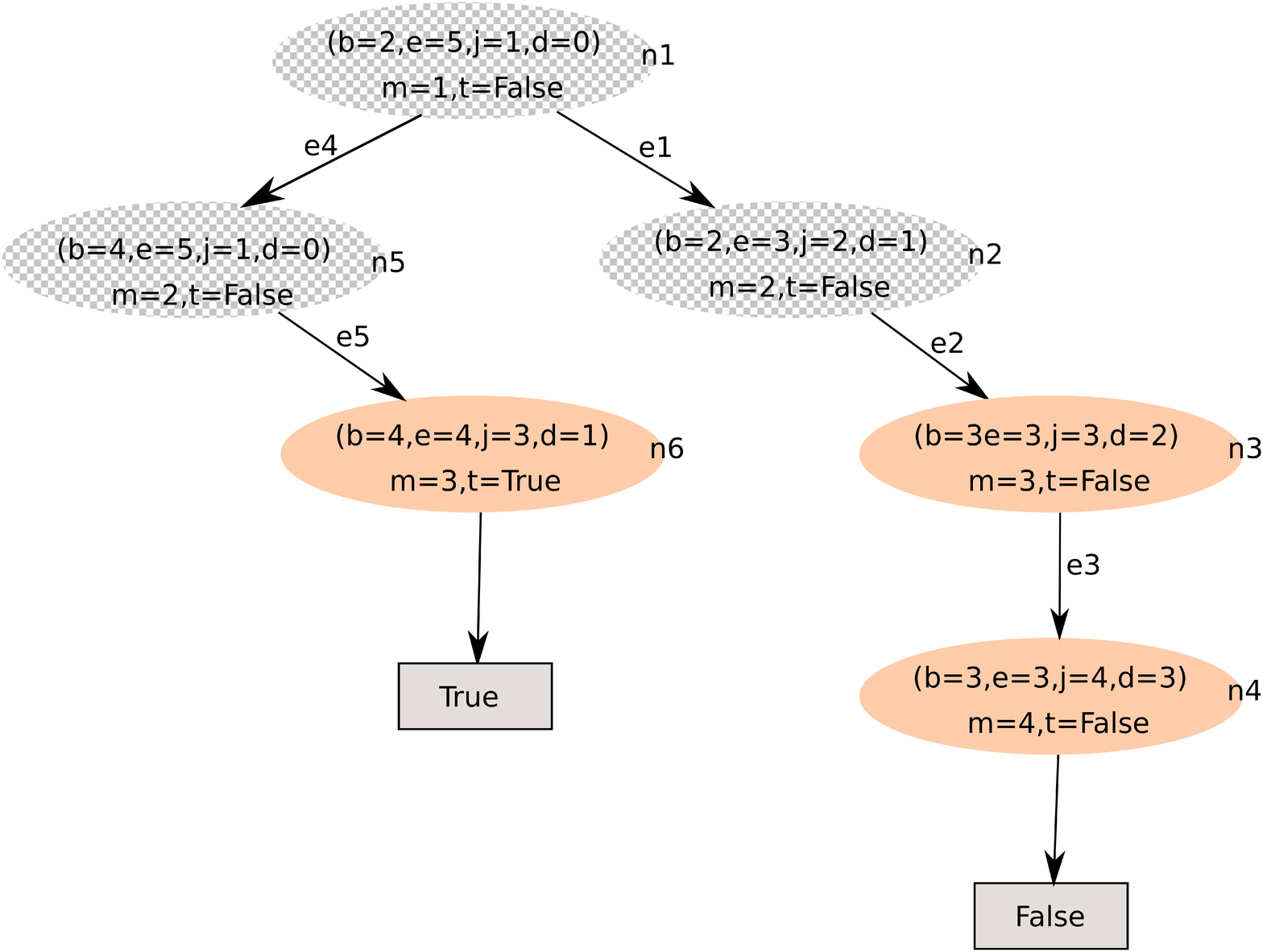}}
	\subfigure[The memoized call graph $v$ after processing the itemset $D_3=abdf$ over the dataset $D$]{\includegraphics[scale=0.13]{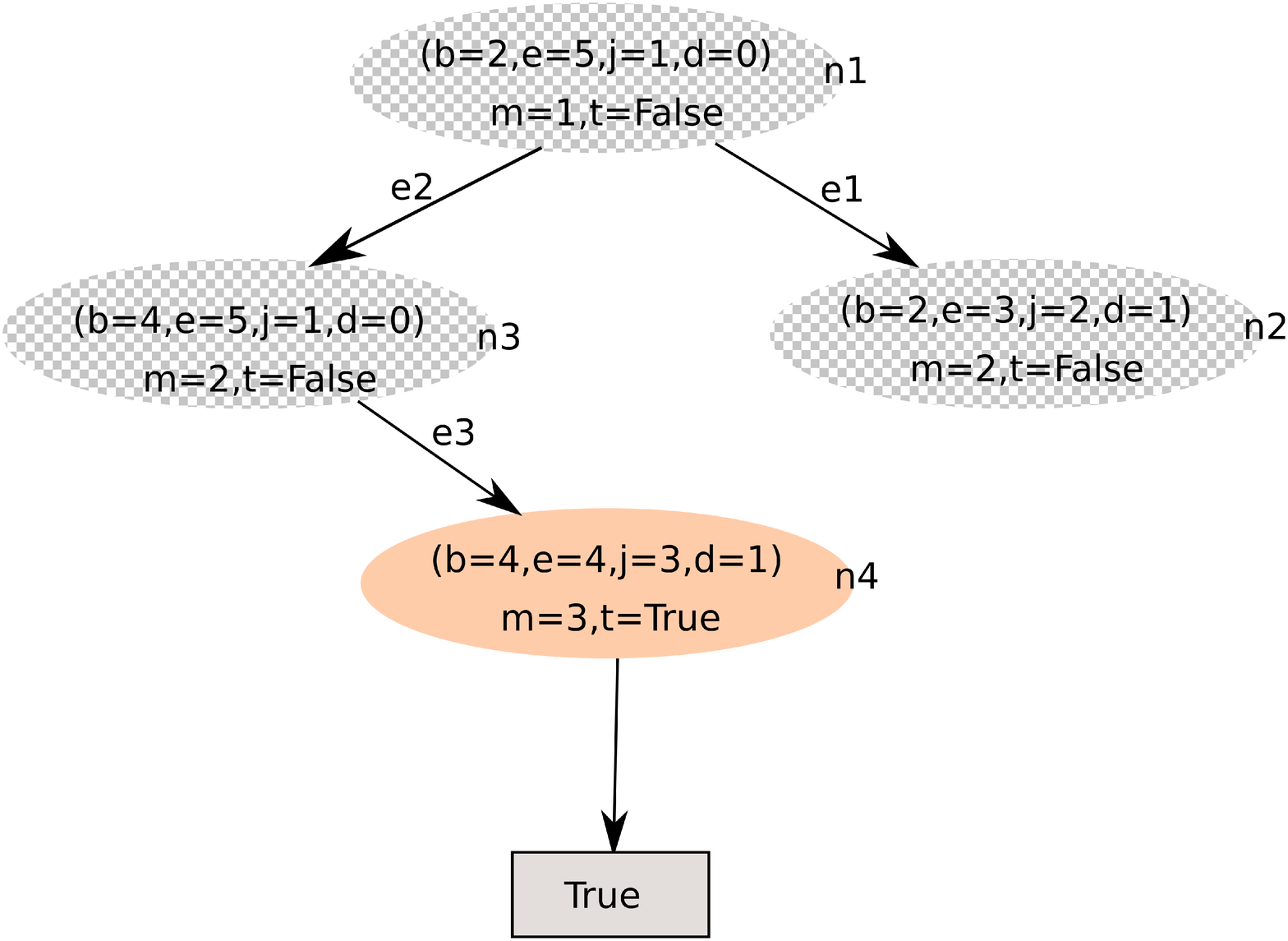}}
	\subfigure[The memoized call graph $v$ after processing the itemset $D_4=bd$ over the dataset $D$]{\includegraphics[scale=0.13]{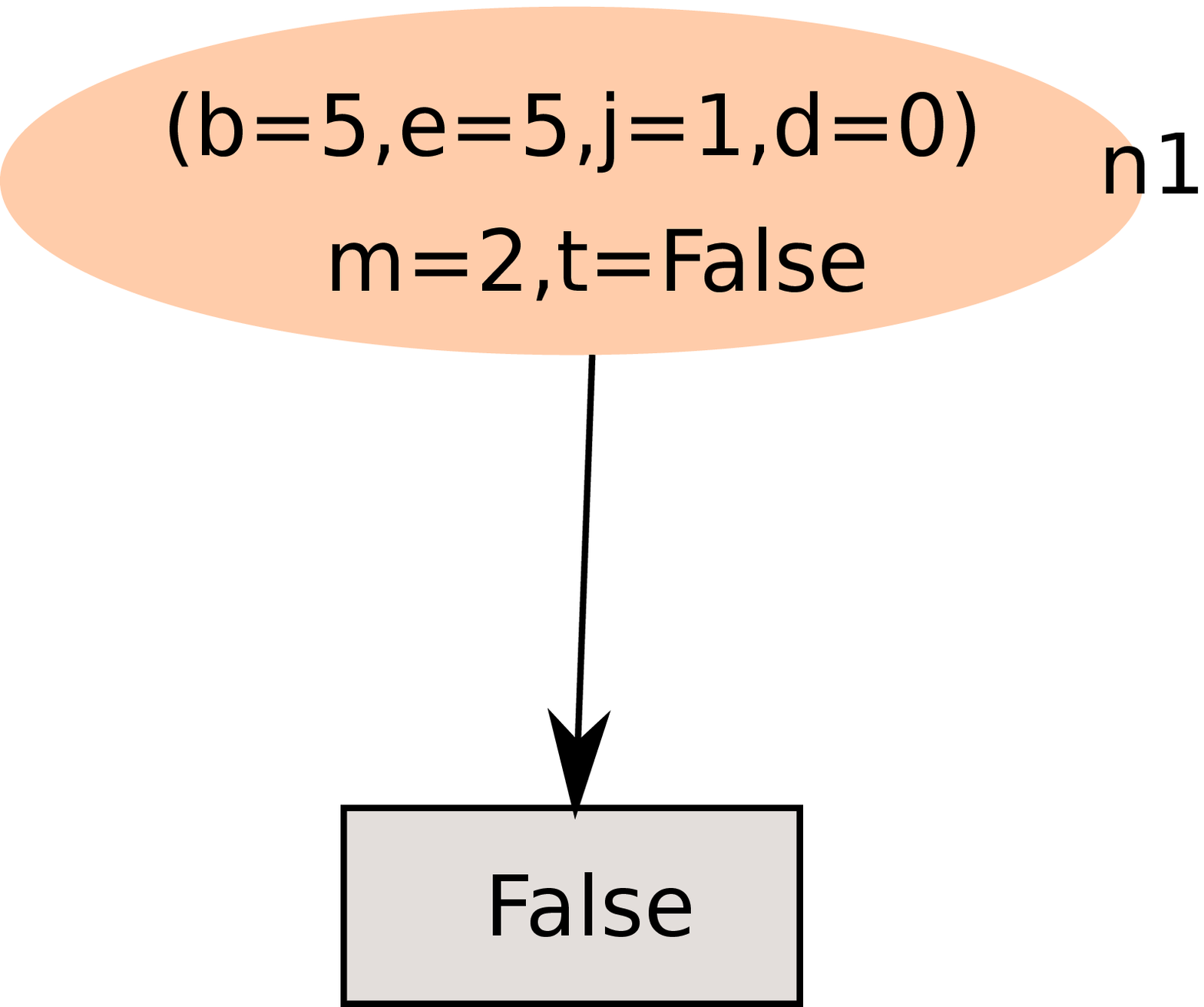}}
	\caption{ This figure presents the evaluation of the memoized version of AMS-Lex over the same dataset $D$ as presented in Figure~\ref{fig:algo:ams-lex:example}. Here we show exactly which parts of the graph are memoized --- the shaded nodes. Each sub-figure shows the memoized call graphs $v_k$ as per Algorithm~\ref{algo:AMS-Lex-Memoized} after processing every itemset $D_k$ from the dataset $D = \{ D_1 = abc, D_2 = abde, D_3 = abdf, D_4 = bd, D_5 = c \}$. All graph nodes and edges are labelled in the order of executions --- first is $n_1$, then $n_2$, then $n_3$, etc. The solid nodes in the graphs are evaluated using Algorithm~\ref{algo:Contains-Subset-Of} and the shaded nodes are memoized. Further explanation of these call graphs can be found in Section~\ref{sec:memoized:algo:example}. }
	\label{fig:algo:memoized:example}
\end{figure}

\subsubsection{Example}
\label{sec:memoized:algo:example}

The sample dataset that the memoized algorithm is evaluated on in Figure~\ref{fig:algo:memoized:example} is the same as the dataset that AMS-Lex is evaluated on in Figure~\ref{fig:algo:ams-lex:example}. The call graphs in Figure~\ref{fig:algo:memoized:example} present visually exactly which parts of the call graphs are memoized --- the shaded nodes --- by keeping track of the memoized call graph --- variable $v$ in Algorithm~\ref{algo:AMS-Lex-Memoized}. 

We see that in general, the memoized call graph of an itemset $D_i$ could be used when processing itemset $D_{i+x}$ for any integer $x > 0$. In our presented example in Figure~\ref{fig:algo:memoized:example} we see that we use part of the memoized call graph from $D_1$ when processing $D_2$ and $D_3$; this happens because $D_1$, $D_2$ and $D_3$ share the non-empty common prefix $ab$.

In Figure~\ref{fig:algo:memoized:example} (c) we see exactly that a subgraph of $v$ gets reset to $null$ --- line~\ref{algo:Contains-Subset-Of-Memoized:v_null} in Algorithm~\ref{algo:Contains-Subset-Of-Memoized}. That is when $D_3$ is processed, the memoized nodes $n_3$ and $n_4$ from Figure~\ref{fig:algo:memoized:example} (b) are removed in the (new) memoized call graph from Figure~\ref{fig:algo:memoized:example} (c).

\subsection{Complexity Analysis}

\paragraph{Worst Case Time Complexity}
It is easy to see that in the worst case (when no two itemsets have a common prefix), the complexity of our algorithm is equal to that of AMS-Lex, that is  $O(N^{2} / \log(N))$, where $N$ is the sum of the cardinalities of all itemsets in the input dataset.

\paragraph{Runtime Comparison to AMS-Lex}
Our algorithm's run time is clearly bounded above by the time required by AMS-Lex. Moreover, as the number of common prefixes among the $itemsets$ increases, the faster (comparatively) our algorithm becomes. Essentially by executing Contains-Subset-Of fewer times, we save run time consumed by the low level searching routines $NextItem$, $NextEndRange$, and $NextBeginRange$ which are the bottleneck of the AMS-Lex algorithm as per \cite{BayardoPanda11}.

\paragraph{Space Complexity}
In addition to the memory required by AMS-Lex, Algorithm~\ref{algo:AMS-Lex} stores (part of) the call graph of Contains-Subset-Of. Clearly the size of the call graph is bounded above by the size of the input, denoted as $N$. Since only the required portion of the call graph, as defined by Corollary~\ref{MartinCorollary}, is stored in practice, the extra space required is commonly much less than the size of the input.

\subsection{Implementation Details}
\label{ImplDetails}
We implemented our algorithm as a modification to the publicly available implementation\footnote{\url{https://code.google.com/p/google-extremal-sets/}} of the AMS-Lex algorithm, only introducing
the memoization described in Algorithm~\ref{algo:AMS-Lex-Memoized}. We regard this as valuable since it allows us to directly measure the improvement in performance resulting from memoization.

\section{A Parallel Algorithm for Identifying Extremal Sets}
\label{sec:parallel}

We use the complexity analysis of the function AMS-Lex~\cite{BayardoPanda11} to identify the bottleneck of the existing algorithm. In the worst case, finding all proper prefix subsumed itemsets takes $ \BigO{ N } $ computational steps and finding the remaining non-minimal itemsets takes $ \BigO{ N^2 / log(N) } > \BigO{ N }$, where $N$ is the size of the input. Consequently, the novel work presented in this section is a parallel algorithm that finds the non-proper prefix subsumed itemsets of $D$, i.e. we present a parallel implementation of the function Get-Minimal-Itemsets-Lex from Algorithm~\ref{algo:AMS-Lex}.

\subsection{Observation}
The first observation we make is that the pseudo code of the function Contains-Subset-Of, presented in Algorithm~\ref{algo:Contains-Subset-Of} that is a reproduction of Contains-Subset-Of~\cite{BayardoPanda11}, does not modify the input dataset $D$. Hence, this makes the algorithm of finding all minimal itemsets within $D$ embarrassingly parallel.


\subsection{Algorithm}

\begin{algorithm} [t]
\SetAlgoNoLine
\KwIn{ Dataset $D = \{ D_1, D_2, \dots, D_{n} \}$ and the degree of parallelism $P$. }
\KwOut{ The minimal itemsets within the dataset $D$. i.e. $Min(D)$. }
\SetKwProg{func}{Function}{}{}
\SetKwFunction{findMinLex}{Get-Minimal-Itemsets-Lex-Parallel}
\SetKwFunction{threadFunctor}{Thread-Functor}
\func{\findMinLex{$\mathbf{dataset}$ $D$, $\mathbf{integer}$ $P$}}
{
	\nl $atomic<bool> is\_min[r] \longleftarrow \{true, true, \dots, true\}$\;
	\tcc{ atomic boolean variables. }
	
	\tcc{ Find itemsets subsumed by proper prefix. }
	\nl $S \longleftarrow D[1]$\;
	
	\nl \For{$i = 2$ \KwTo $n$}
	{
		\nl \If{$ |S| \leq |D[i]| \And D[i][1 : |S|] = S $}
		{
			\tcc{ $S$ is a proper prefix of $D[i]$. }
			\nl $is\_min[i] \longleftrightarrow false$\;
		}
		\Else
		{
			\nl $S \longleftarrow D[i]$\;
		}
	}

	\tcc{ Find itemsets subsumed by non-proper prefix using $P$ parallel threads. }
	\nl $atomic<int> index \longleftarrow 1$\;
	\tcc{ the index that is to be processed next. }
	
	\nl start $P$ parallel instances of \threadFunctor$( D, index, is\_min )$\;
	
	\nl wait for all $P$ instances to finish working\;
	
	\nl \KwRet $ \{ D_i \in D$ $\lvert$ $is\_min[i] == true \}$\;
}

\func{\threadFunctor{$\mathbf{dataset}$ $D$,  $\mathbf{atomic<integer>}$ $index$, $\mathbf{atomic<bool>}$ $m[r]$}}
{
	\nl $i \longleftarrow$ fetch-and-increment$(index)$\;
	\tcc{an atomic operation}
	
	\nl \While{$i \leq n$}
	{
		\tcc{ It is safe to invoke the function $\containsSubsetOf$ from multiple threads at the same time as it requires only read-only access to the dataset $D$. }
		\nl \If{\containsSubsetOf$( D[i+1 : n], D[i], 1, 0 )$ \tcc{ as per Algorithm~\ref{algo:Contains-Subset-Of} } }
		{
			\tcc{ mark the $i$-th itemset as non-minimal because the dataset $D$ contains a proper subset of the itemset $D[i]$.}
			\nl $m[i] \longleftarrow false$\;
			\tcc{ atomically setting the $i$-th boolean value.}
		}
		
		\nl $i \longleftarrow $ fetch-and-increment$(index)$\;
	}
}
\caption{Pseudo code for finding the minimal itemsets $M$ of the input dataset $D = \{  D_1, D_2, \dots, D_{n} \}$ using $P$ threads. We present a subroutine Find-Min-Lex which identifies the minimal itemsets of $D$ using $P$ parallel threads. It is important to note that in the Thread-Functor subroutine the variables $index$ and $is\_min$ are passed by reference, meaning that they are shared between threads.}
\label{algo:AMS-Lex-Parallel}
\end{algorithm}

The pseudo code for our parallel algorithm of finding the minimal itemsets within a lexicographically ordered dataset is presented in Algorithm~\ref{algo:AMS-Lex-Parallel}.

\paragraph{Entry Point}
We first mark every itemset within the dataset $D$ as minimal. Next, we mark all itemsets as not minimal for which there exists a proper prefix subsumed itemset within the dataset. We then start $P$ parallel instances of the thread functor whose job is to mark itemsets as non-minimal for which there exists a non-prefix (lexicographically larger) subsumed itemset. 

\paragraph{Thread Functor}
All of the parallel instances of the Thread-Functor function share a common integer variable $index$ which points to the next unprocessed itemset $D[index] \in D$ within the datasets starting at $1$. To process the itemset $D[index]$ means to check if there exists a non-prefix subsumed within $D$ of $D[index]$. We begin by atomically assigning the current value of $index$ to the variable $i$ and incrementing $index$; ensuring that every itemset in $D$ will be processed exactly once by some Thread-Functor. We then use the function Contains-Subset-Of from Algorithm~\ref{algo:Contains-Subset-Of} to check if a subset of $D[i]$ is found. Finally, we try to take a new unprocessed itemset from $D$ and process it in the same manner.

\subsection{Complexity}
Here we give the worst case time and space complexity of the functions presented in Algorithm~\ref{algo:AMS-Lex-Parallel}. From Bayardo~and~Panda~\cite{BayardoPanda11}'s complexity analysis of AMS-Lex we know that the worst case time complexity of AMS-Lex is equal to  $\BigO{ N }$ to identify the prefix subsumed itemsets and additional $\BigO{ N^2 / log(N) }$ to find the non-prefix subsumed ones; recall that $N$ denotes the sum of the cardinalities of all the sets in the input dataset $D$. Since in this section we showed that, the function Contains-Subset-Of requires only read-only access to the dataset $D$ and we have $P$ threads at our disposal we deduce that worst case execution time of the function Get-Minimal-Itemsets-Lex-Parallel is $\BigO{ N } + \BigO{ N^2 / (log(N) \times P) } = \BigO{ N + N^2 / (log(N) \times P) } $; note that $1 \leq P \leq n$. As for the space complexity of the Get-Minimal-Itemsets-Lex-Parallel algorithm it is equal to that of Get-Minimal-Itemsets-Lex~\cite{BayardoPanda11} which is proportional to the size of the input, i.e. $\BigO{ N }$.

\section{Experiments}
\label{sec:experiments}

Here we describe the experimental comparison of our algorithm with Bayardo and Panda's algorithm AMS-Lex for identifying the minimal itemsets within a dataset. We measure execution time speedup as the ratio of AMS-Lex algorithm execution time divided by our algorithm's execution time. Hence, a speedup of $2$ means that our algorithm executed in half the time, and a value of $1$ means that both algorithms have the same execution time. For every input, we also measure the total number of calls that each algorithm made to the subroutines $NextBeginRange$ and $NextEndRange$, because as described in \cite{BayardoPanda11}, these subroutines are the bottleneck of the AMS-Lex algorithm. In our experimental evaluation we provide a link between the decrease in the number of range searches performed by our algorithm in comparison to AMS-Lex and the relative to AMS-Lex execution time speedup.

Although not presented below, we also conducted experiments with the Bayardo and Panda's AMS-Card Algorithm on all of the data and it performed slower on all cases, compared to the AMS-Lex algorithm. That is expected, as stated by Bayardo and Panda \cite{BayardoPanda11}, the cardinality approach is faster then the lexicographic one mostly primarily in very obscure and rare cardinality distributions. Furthermore, the goal of this paper is to present faster than AMS-Lex methods of finding extremal sets that are based on Pritchard's lexicographic subsumption property from Theorem~\ref{th:lex}.

\subsection{Experimental Setup}
For all of our experiments we used a machine with four Intel Xeon CPU E7- 4820, each with eight cores, clocked at $2.00GHz$, a third level cache size of $18MB$ and $128GB$ of main memory. Note that our experiments investigate the case when the entire input fits in main memory. We used uniform random data as well as publicly available data as input to evaluate our two new algorithms and AMS-Lex. All of the results presented below are averaged over $3$ different runs.

\subsection{Real-World Data}

\begin{figure} [t]
	\centering
	\includegraphics[scale=0.5]{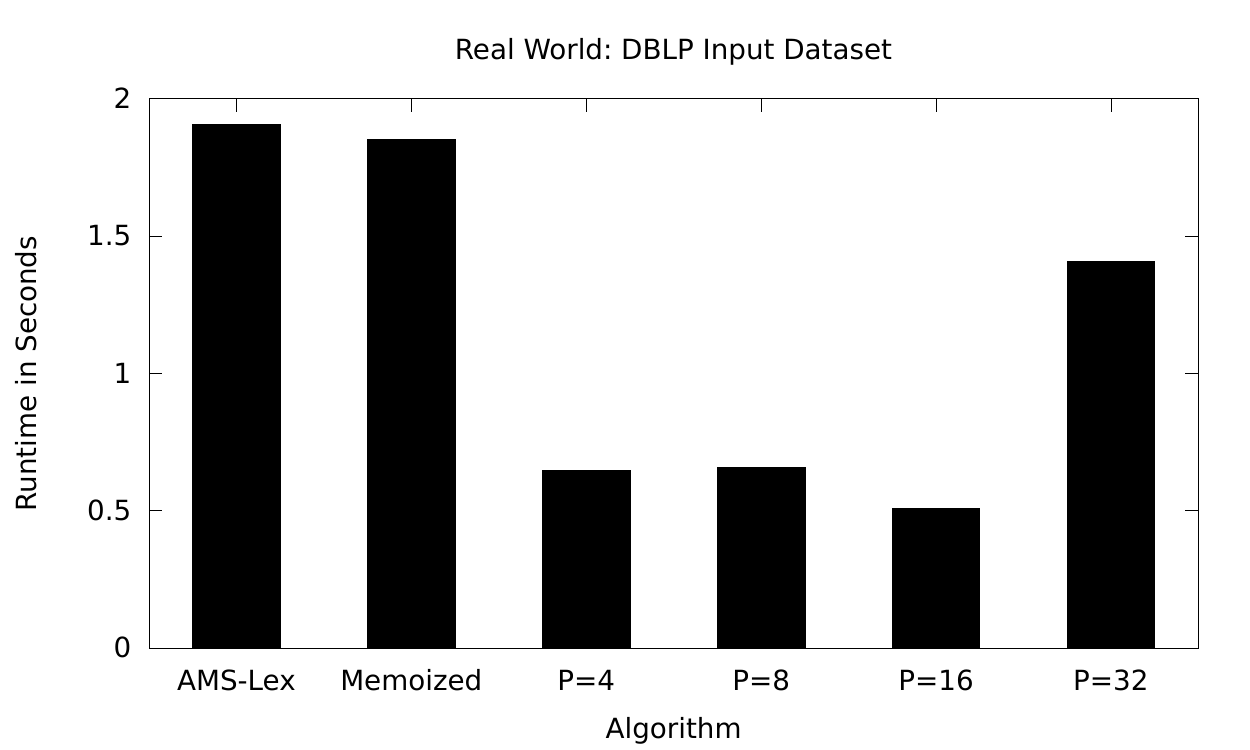}
	\includegraphics[scale=0.5]{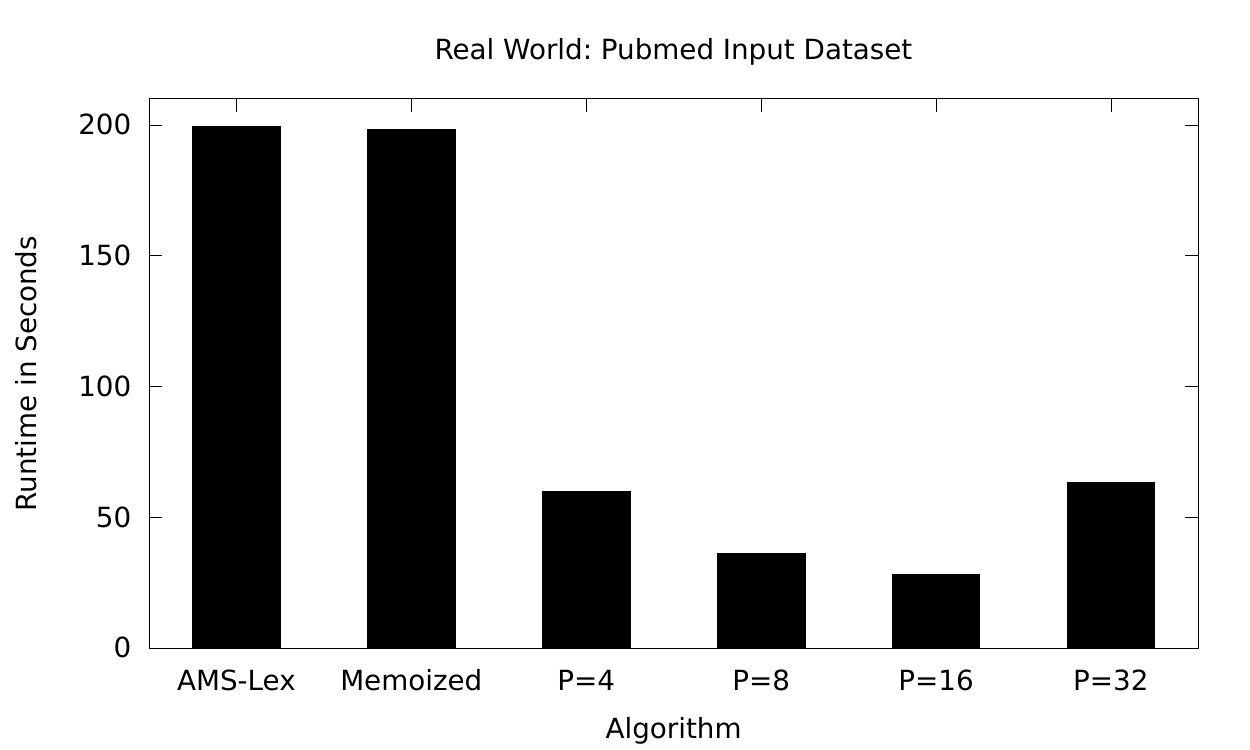}
	\includegraphics[scale=0.5]{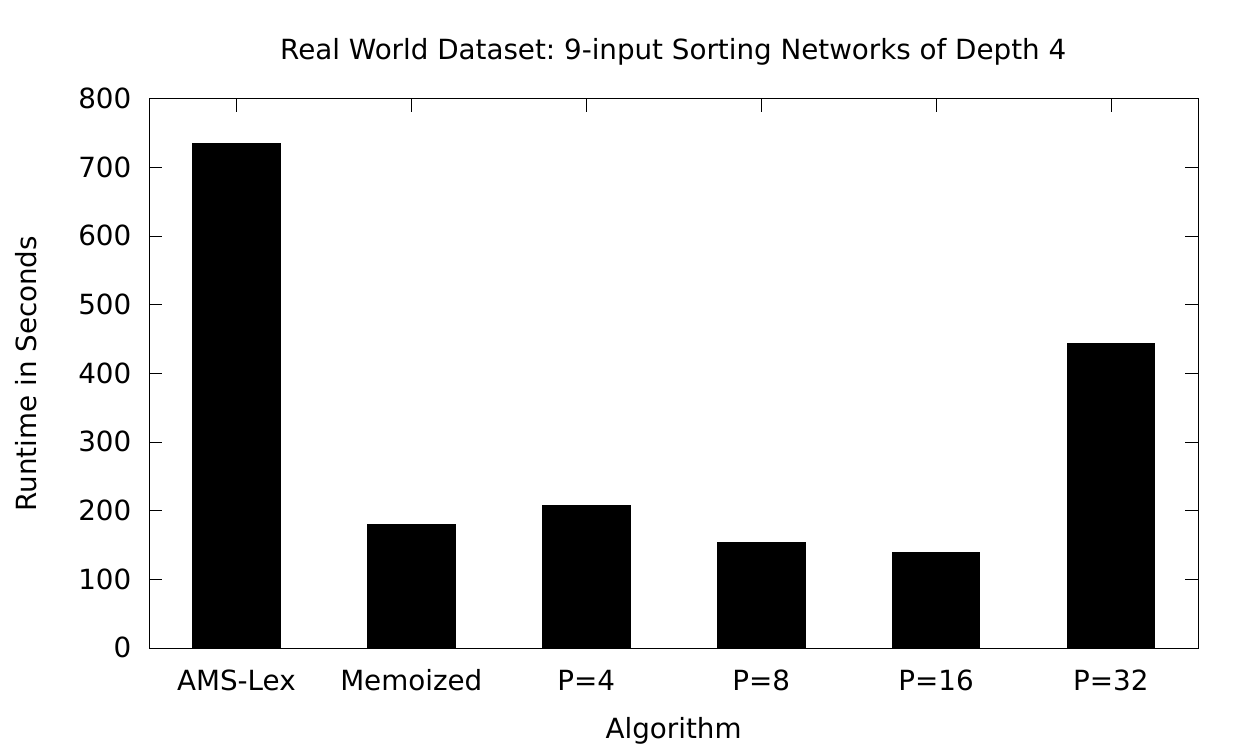}
	\includegraphics[scale=0.5]{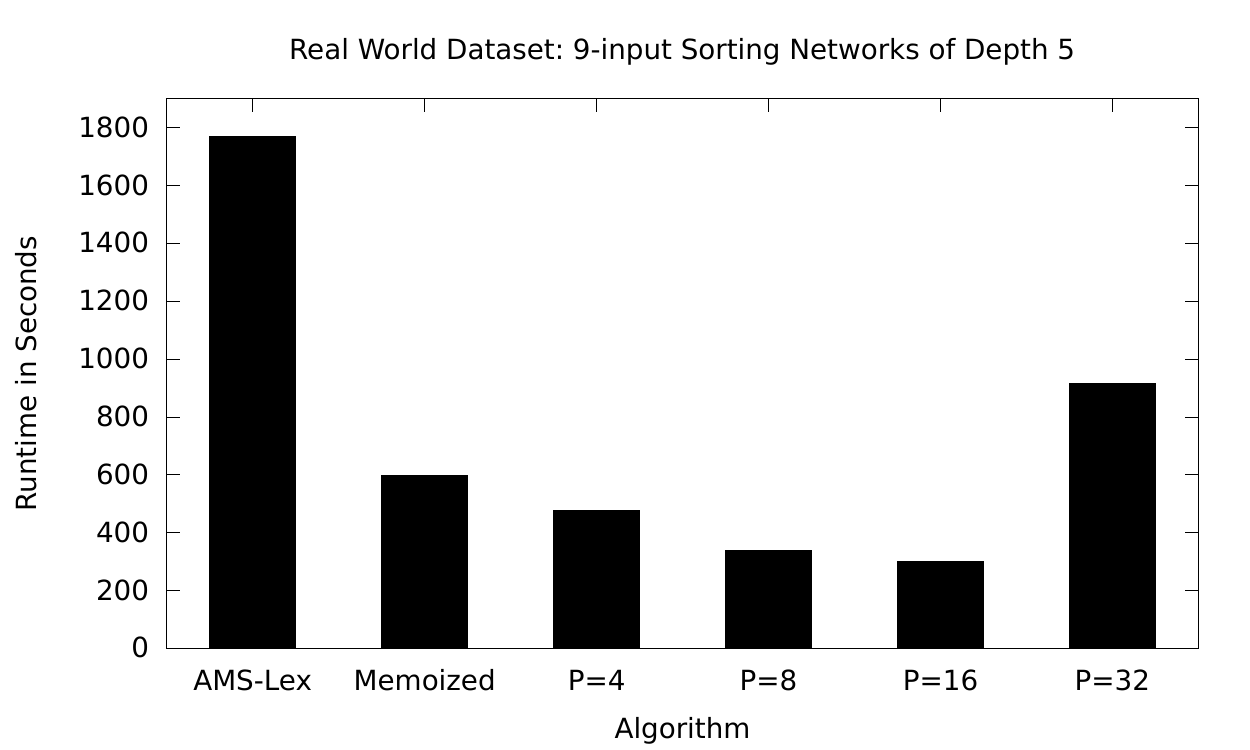}
	\caption{Experimental results using real world datasets, comparing AMS-Lex with the memoized (section~\ref{sec:memoized}) and parallel (section~\ref{sec:parallel}) approach for finding the minimal itemsets within a dataset. For these results we have used a machine with $32$ physical cores and used parallelism factors $P=4,8,16$ and $32$ for our parallel modification of AMS-Lex.}
	\label{fig:exp:real}
\end{figure}

A summary of the conducted experiments using real-world input datasets is presented in Figure~\ref{fig:exp:real}. We have evaluated the AMS-Lex algorithm, our memoized approach and the parallel method using different degrees of parallelism over the real-world datasets:

\begin{itemize}
\item \textbf{PubMed} dataset represents significant terms in the PubMed abstract. It consists of $8$ million itemsets stored in a $2GB$ file.
\item \textbf{DBLP} dataset consists of $1$ million itemsets and is used in the area of similarity joins. The file size is $50MB$.
\item \textbf{SN\_9\_4} dataset consists of $2$ million itemsets with an average size of $30.3$ and an alphabet size of $2^9$. This data is derived from the domain of $9$-input sorting networks by generating all non maximal networks of depth $4$ . The file size is $252MB$.
\item \textbf{SN\_9\_5} dataset consists of $7.5$ million itemsets with an average size of $18.1$ and an alphabet size of $2^9$. This data is derived from the domain of $9$-input sorting networks by generating all non maximal networks of depth $5$ by using the minimal ones of depth $4$. The file size is $578MB$.
\end{itemize}

\paragraph{Sorting Networks Datasets}

Here we give explanation on how the datasets $SN\_9\_4$, $SN\_9\_5$ were generated. We refer to the work of Bundala~\textit{et al.}~\cite{BundalaCCSZ14_Optimal_Depth} (Lemma 2 in Section 3.2) about searching for sorting networks of optimal depth. They describe a method of reducing the search space by considering `output-minimal networks' i.e. given a dataset their algorithm needs to identify and consider only the minimal representative itemsets within this dataset. The input dataset $SN\_9\_4$ is generated by applying all maximal network levels to the minimal outputs (itemsets) of networks of depth three; similarly the dataset $SN\_9\_5$ is generated by taking the minimal networks of depth four and applying all maximal network levels. 

The algorithm described in this paper is originally designed to find such output-minimal networks and hence it is aimed at finding the minimal itemsets within a dataset and not the maximal ones as per Bayardo and Panda's approach \cite{BayardoPanda11}. In the background related Section~\ref{sec:background} we describe in detail Bayardo and Panda's AMS-Lex algorithm in terms of finding the minimal itemsets. Bayardo and Panda note that AMS-Lex can be used for finding the minimal and maximal itemsets and that the changes needed to do one or the other are trivial. We chose to work in terms of finding the minimal itemsets within a dataset because our algorithm (and source code) is initially build for tackling the sorting networks related datasets.

\subsubsection{Memoized vs AMS-Lex}

Figure~\ref{fig:exp:real} shows a comparison of the original AMS-Lex and our two modified versions for real world datasets. For the $DPLP$ and $PubMed$ datasets the memoized approach is marginally faster than the AMS-Lex algorithm because there are very few itemset pairs that share a common prefix. On the other hand, for the $SN\_9\_4$ dataset the memoized algorithm is $4.06$ times faster than AMS-Lex; and $2.96$ times faster for the $SN\_9\_5$ dataset. The sorting network input datasets tend to share long common prefixes as the size of the alphabet is very small compared to the size of the input which favours our memoization technique over AMS-Lex. It is important to note that in the sorting network datasets there are no trivially subsumed itemsets.

\subsubsection{Parallel vs AMS-Lex}

Note that our parallel algorithm is executed on a machine with $32$ physical cores and all real-world experimental results are presented in Figure~\ref{fig:exp:real}. For the $DBLP$ dataset we see that the speedup of the parallel algorithm over AMS-Lex is about $3.5$ for degrees of parallelism $P = 4, 8$ and $16$ whereas for $P=32$ we see a reduced speedup. For the $PubMed$ dataset we see substantial speedup for all of the parallelism factors with $P=16$ executing $5.6$ times faster than AMS-Lex. Substantial execution time speedups are evident in the $SN\_9\_4$ and $SN\_9\_5$ datasets both of them peeking at $P=16$ with maximum speedup factors of $5.3$ and $5.9$ respectively. We elaborate more on the explanation of the performance differences between the parallel algorithm and AMS-Lex in Section~\ref{sec:exp:synth:prallel_AMS}. It is important to note that these real-world data execution time speedups are comparatively equal and/or better than the ones that \cite{Fort+13}'s approach achieves over the AMS-Lex algorithm. Hence, we conclude that our parallel version of AMS-Lex is faster than original AMS-Lex on real-world data and competitive with the implementation in \cite{Fort+13}.

\subsection{Synthetic Data}
\label{sec:exp:synth}

\subsubsection{Input Dataset Generation}
We now describe the process of generating random input data using a random data generator program $g(n, d, f_{min})$. The input to the generator is the number of itemsets $n$, the number of distinct items $d$ in the alphabet and the minimal item frequency $f$. Then for each of the $d$ items we choose a frequency $f_i$ from the range $[f_{min}, 1]$ which indicates the number of itemsets which contain this item. Then we insert this item to a set of randomly chosen $ \lfloor f_i \times n \rfloor $ itemsets. Then we use Bayardo and  Panda's open source implementation to sort the input data in the format required by the algorithms. Note that the higher the value of the minimal frequency $f_{min}$ the greater the probability that two itemsets will share a common prefix. We use the value of $f_{min}$ to evaluate our hypothesis that our algorithm is faster than AMS-Lex on inputs consisting of itemsets sharing large common prefixes.

\subsubsection{Memoized vs AMS-Lex}
\label{sec:exp:synth:memo_AMS}

\begin{figure} [t]
	\centering
	\includegraphics[scale=0.5]{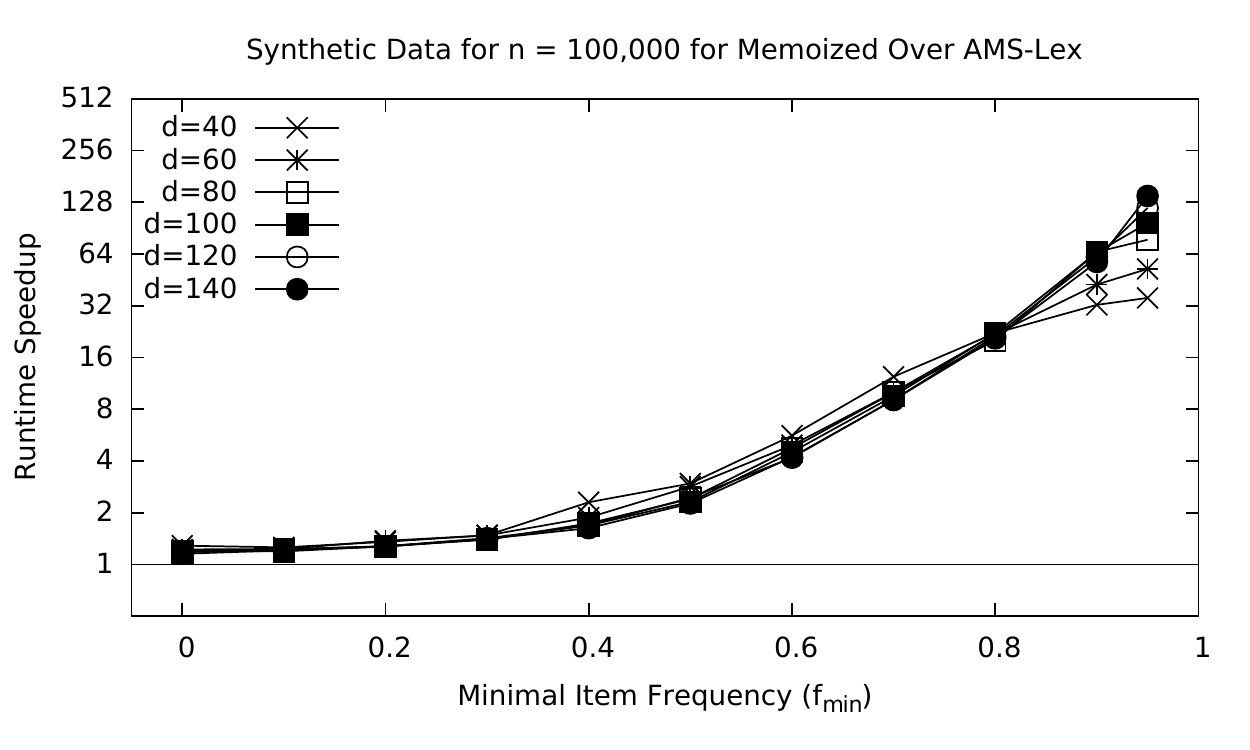}
	\includegraphics[scale=0.5]{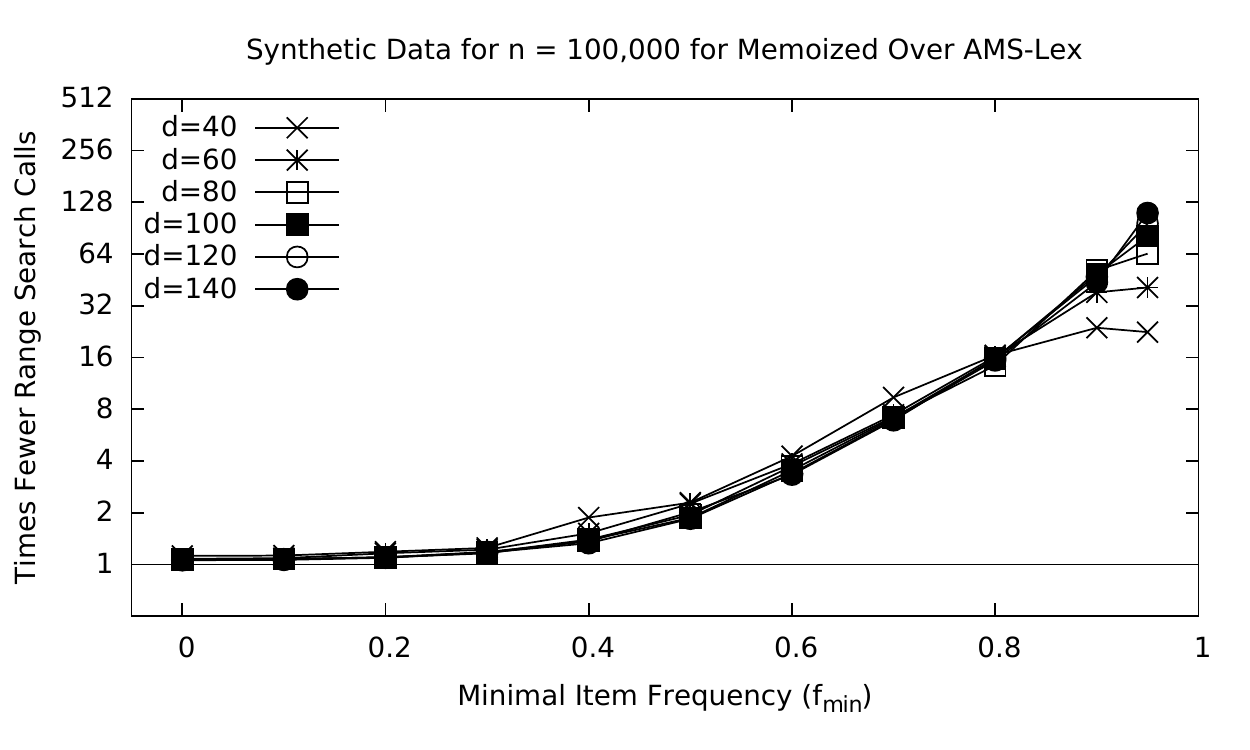}
	\includegraphics[scale=0.5]{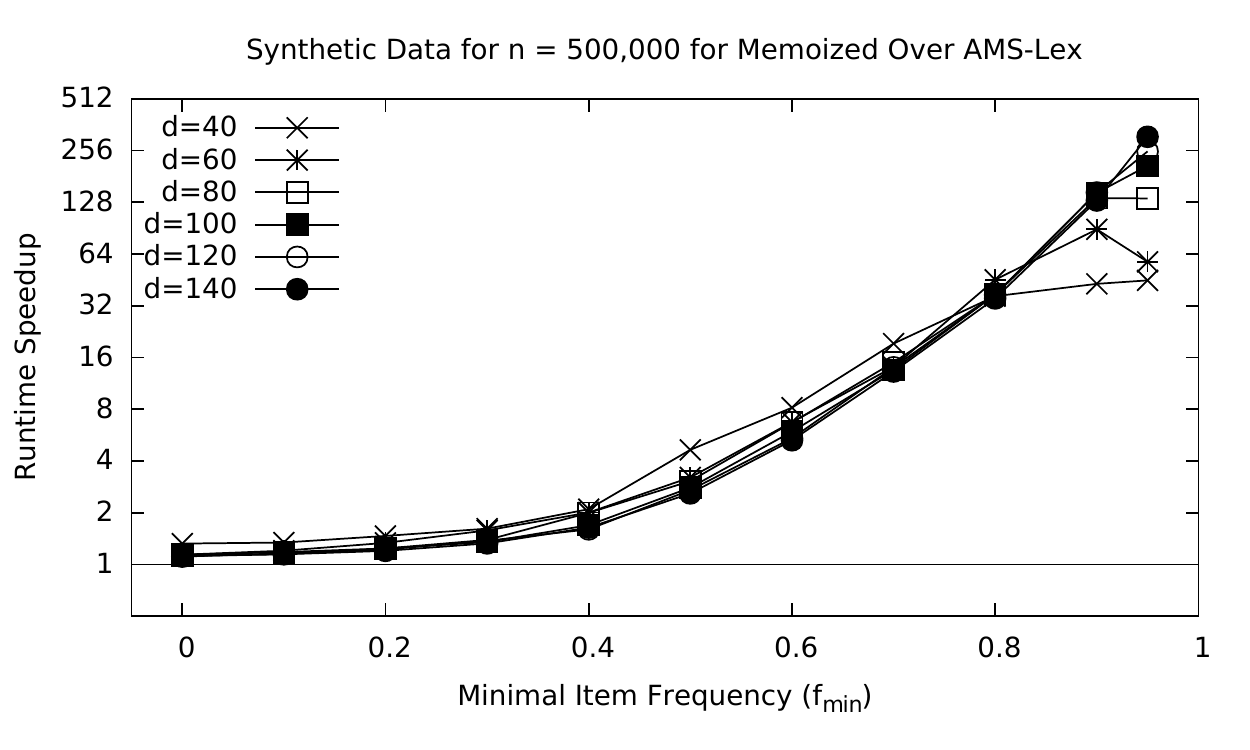}
	\includegraphics[scale=0.5]{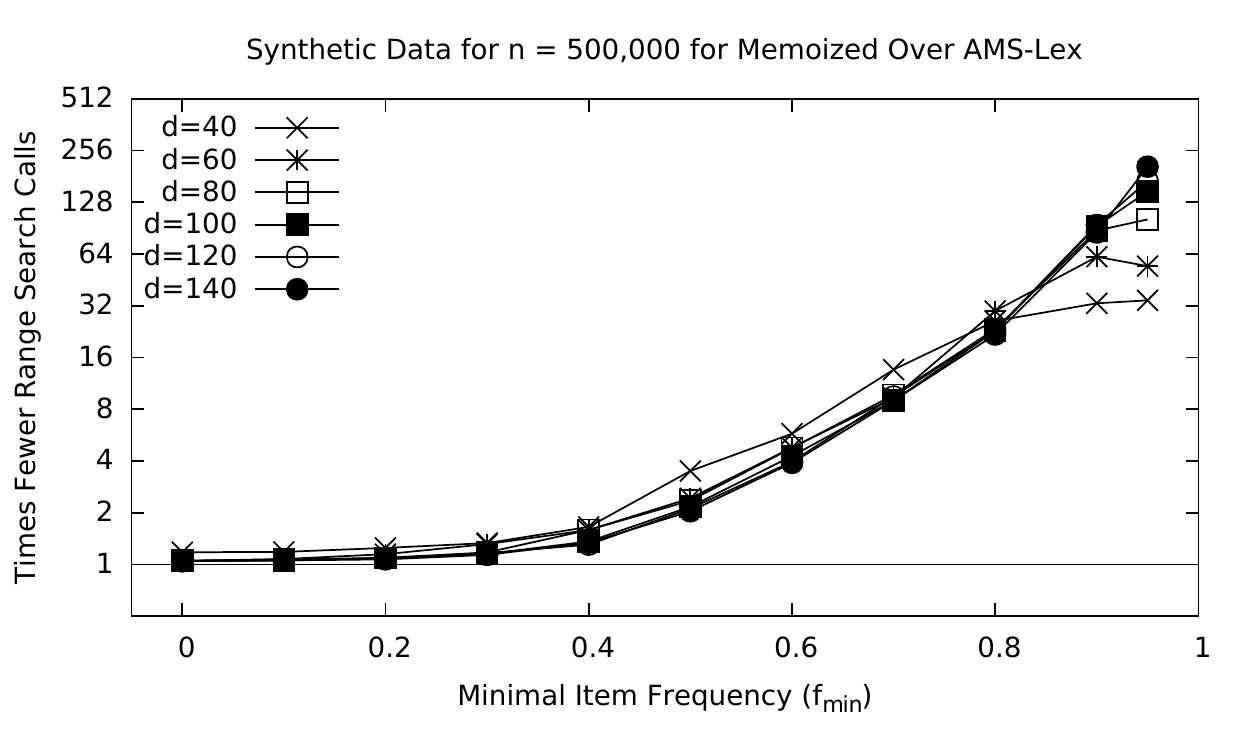}
	\includegraphics[scale=0.5]{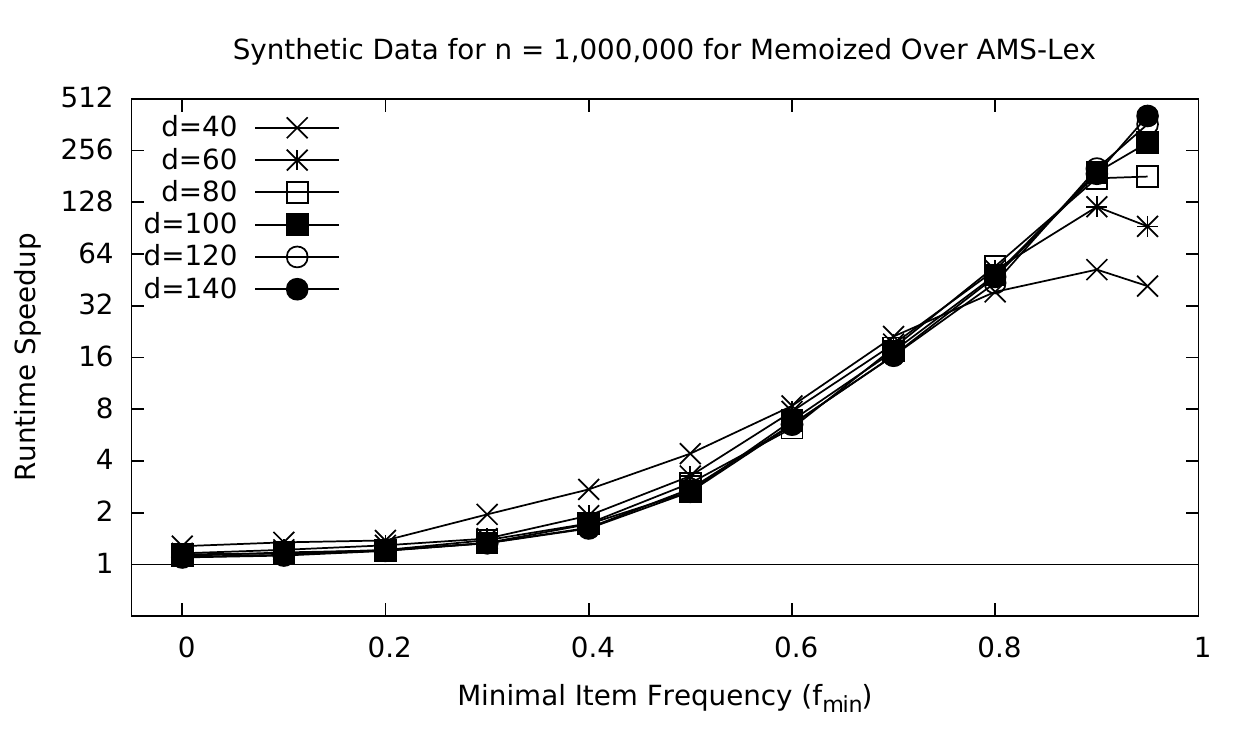}
	\includegraphics[scale=0.5]{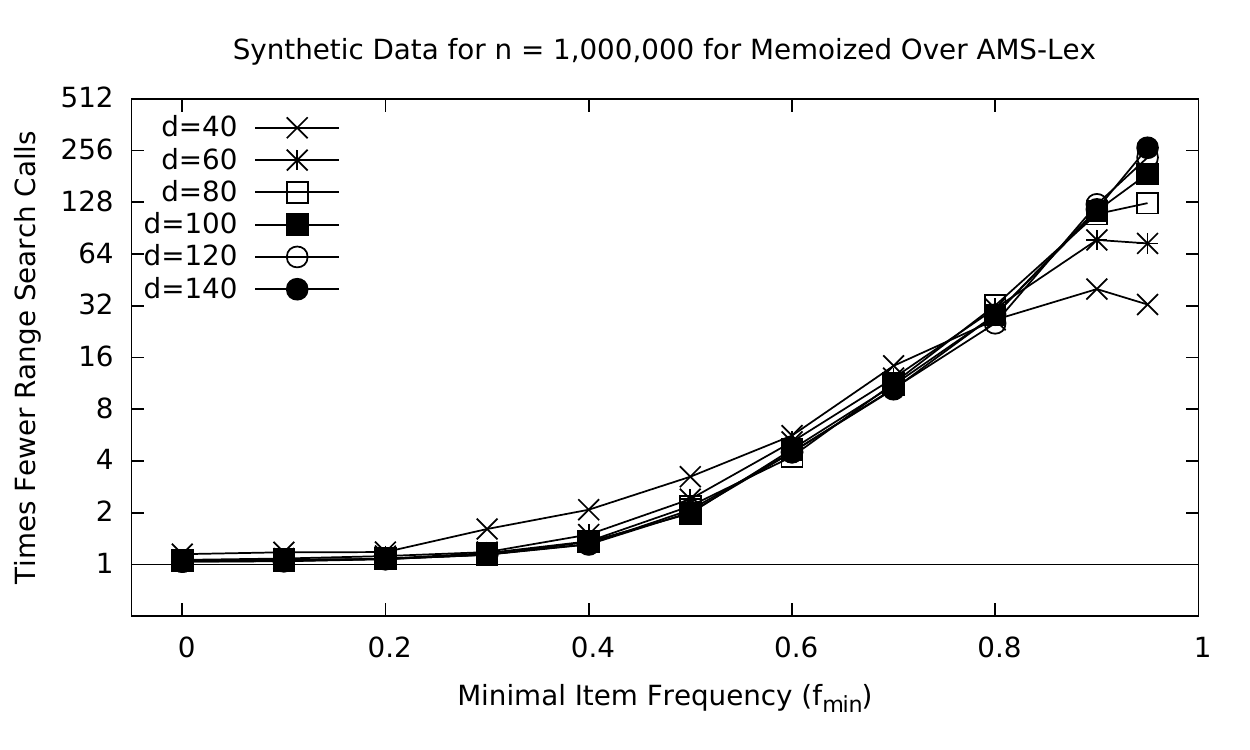}
	\caption{Experimental results using synthetic data for $n = 100\,000$, $n = 500\,000$ and $n = 1\,000\,000$ of comparing our memoized version of AMS-Lex (section~\ref{sec:memoized}) over AMS-Lex for finding the minimal itemsets within a dataset. Here $d$ is the cardinality of the domain of the itemsets. These results show the minimal item frequency ($f_{min}$) described in Section~\ref{sec:experiments} against the resulting execution time speedup as well as the decrease in range search calls of our memoized algorithm compared over AMS-Lex. Note that the y-axis in every graph uses a $\log_2$ scaling for visual clarity of the presented graphs. }
	\label{fig:exp:memoized}
\end{figure}

Figure~\ref{fig:exp:memoized} shows the execution time speedup factor of our memoized algorithm over AMS-Lex for datasets consisting $n = 100\,000$, $n = 500\,000$ and $n = 1\,000\,000$ itemsets with alphabet size of $40$, $60$, $80$, $100$, $120$ and $140$. We notice that as the minimal item frequency increases, the speedup factor increase drastically. The maximum execution time speedup factor of $406$ is achieved by a dataset consisting of $N = 1\,000\,000$ itemsets with alphabet size of $D = 140$ and minimal frequency of $F = 0.95$. We also note that there is an approximately constant correlation between the execution time speedup of our algorithm and the factor of reduction in range search calls. That is an expected correlation because these low level subroutines are described as the bottleneck of AMS-Lex~\cite{BayardoPanda11}. 

In Section~\ref{sec:memoized} we showed that the more common prefixes that itemsets have, i.e. as $f_{min}$ increases and we keep $n$ and $d$ fixed, the bigger the expected speedup factor, which is experimentally verified by this figure. We note that fixing the size of the alphabet $d$ and the minimal item frequency $f_{min}$, in Figure~\ref{fig:exp:memoized} we see that as the number of itemsets $n$ increases, the execution time speedup of the memoized algorithm over AMS-Lex increases. Also, if we fix $n$ and $f_{min}$ we see that as $d$ increases the execution time speedup is non-decreasing in all of the conducted experiments.

\begin{figure} [t]
	\centering
	\includegraphics[scale=0.7]{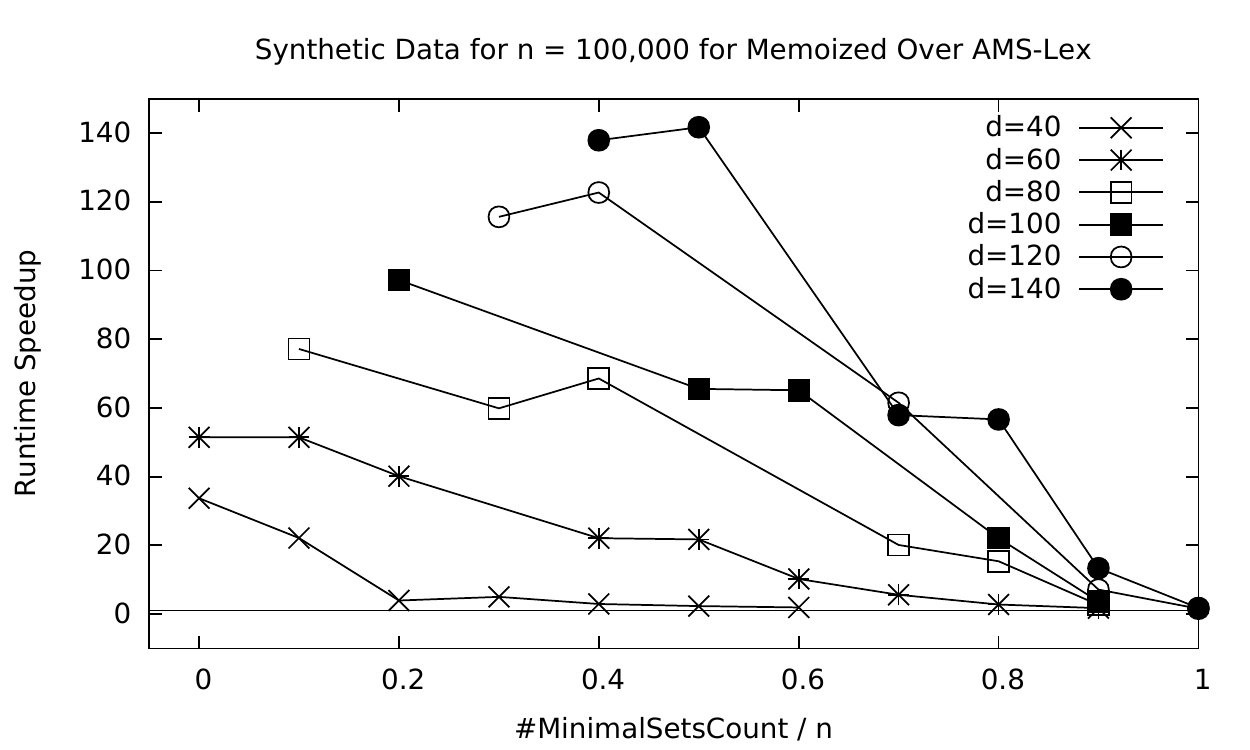}
	\includegraphics[scale=0.7]{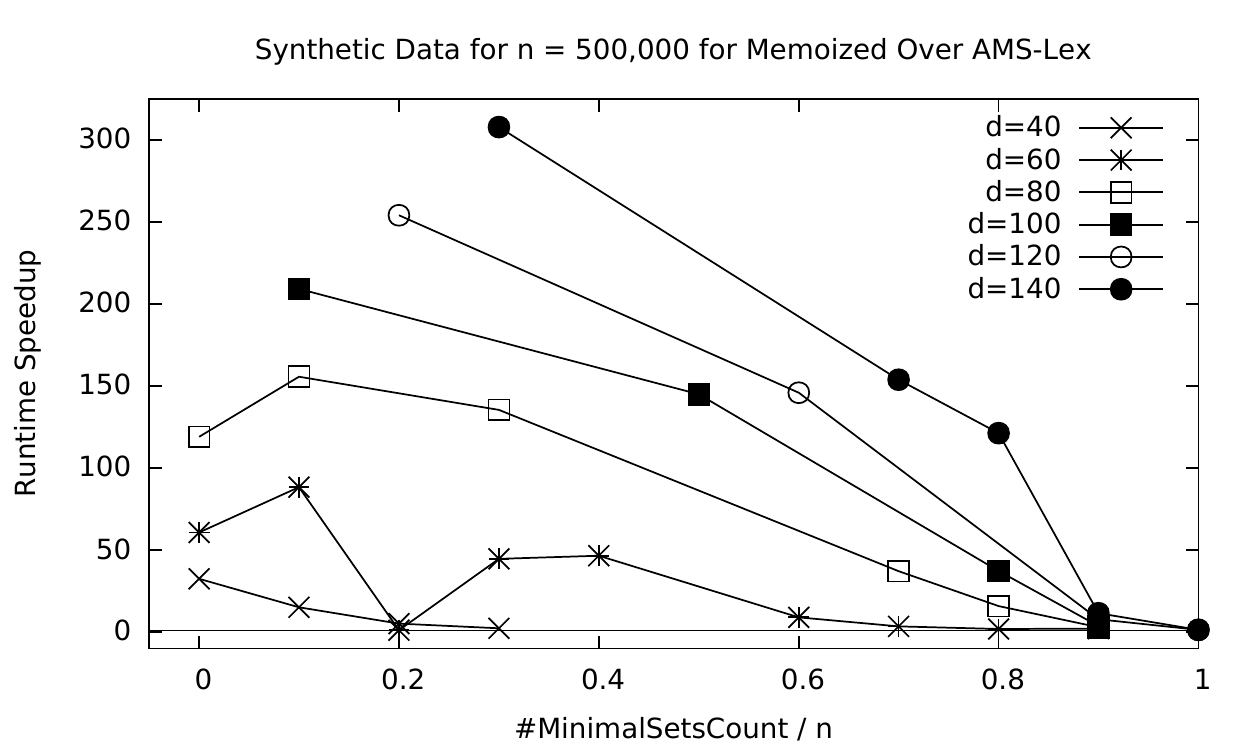}
	\includegraphics[scale=0.7]{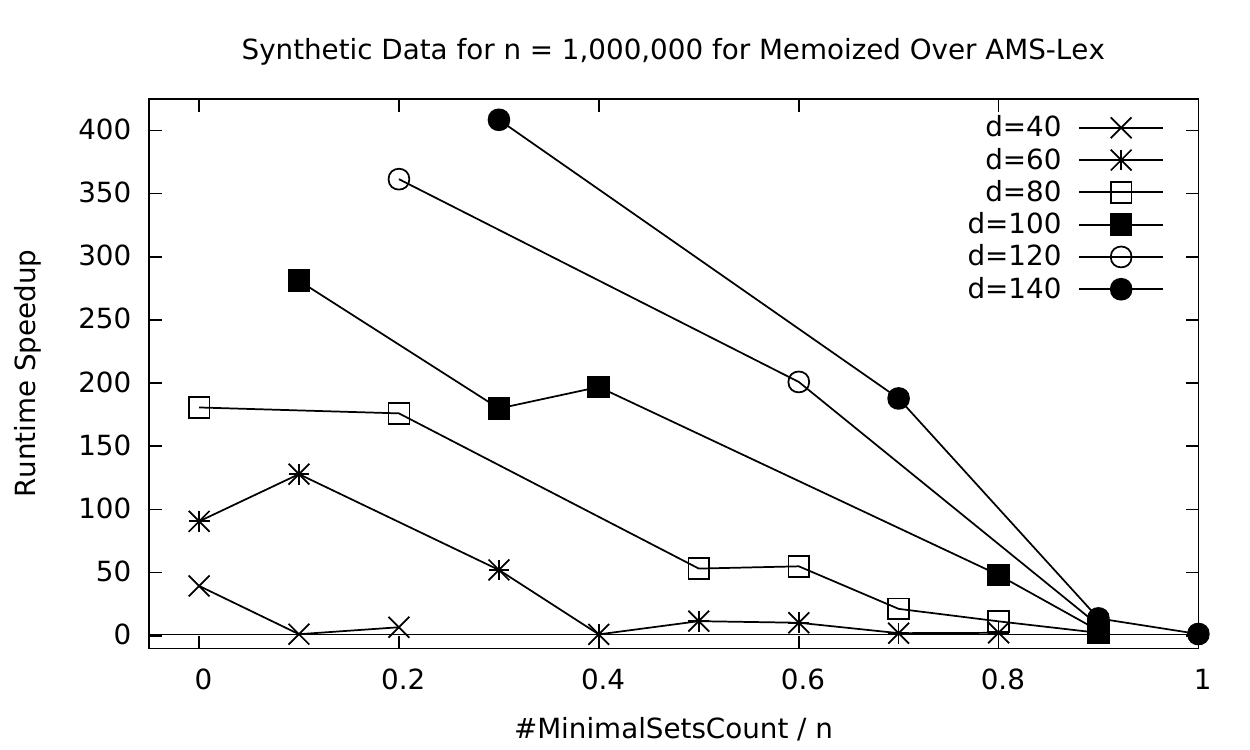}
	\caption{Experimental results using synthetic data for $n = 100\,000$, $n = 500\,000$ and $n = 1\,000\,000$ of comparing our memoized version of AMS-Lex (section~\ref{sec:memoized}) over AMS-Lex for finding the minimal itemsets within a dataset. Here $d$ is the cardinality of the alphabet. These results show the number of minimal itemsets against the resulting execution time speedup of our memoized algorithm compared to AMS-Lex.}
	\label{fig:exp:memoized:res_count}
\end{figure}

Another interesting summary of our experiments is shown in Figure~\ref{fig:exp:memoized:res_count} which gives the execution time speedup with respect to the cardinality of the resulting minimal itemsets by presenting three different graphs for $n = 100\,000$, $n = 500\,000$ and $n = 1\,000\,000$. Our first impression is that all of the graphs look very similar to each other besides the scale of the execution time speedup access. Our second observation shows that the largest speedups are almost always achieved at the smallest resulting minimal sets count for every $d$ and $n$. Moreover, as $d$ increases the absolute maximum speedup increases as well and all speedups tend to $0$ when the size of the result is close to the size of the input ($0.9$ to $1.0$). Reading the graphs in Figures~\ref{fig:exp:memoized}~\ref{fig:exp:memoized:res_count} we deduce that there is a correlation between the minimal item frequency $f_{min}$ and the resulting minimal sets count --- as $f_{min}$ increases the number of minimal sets decreases. Hence, in Figure~\ref{fig:exp:memoized:res_count} we observe that as the number of minimal sets increases the speedup decreases; and in Figure~\ref{fig:exp:memoized} we see that as $f_{min}$ increases the speedup increases.

\subsubsection{Parallel vs AMS-Lex}
\label{sec:exp:synth:prallel_AMS}

\begin{figure} [t]
	\centering
	\includegraphics[scale=0.5]{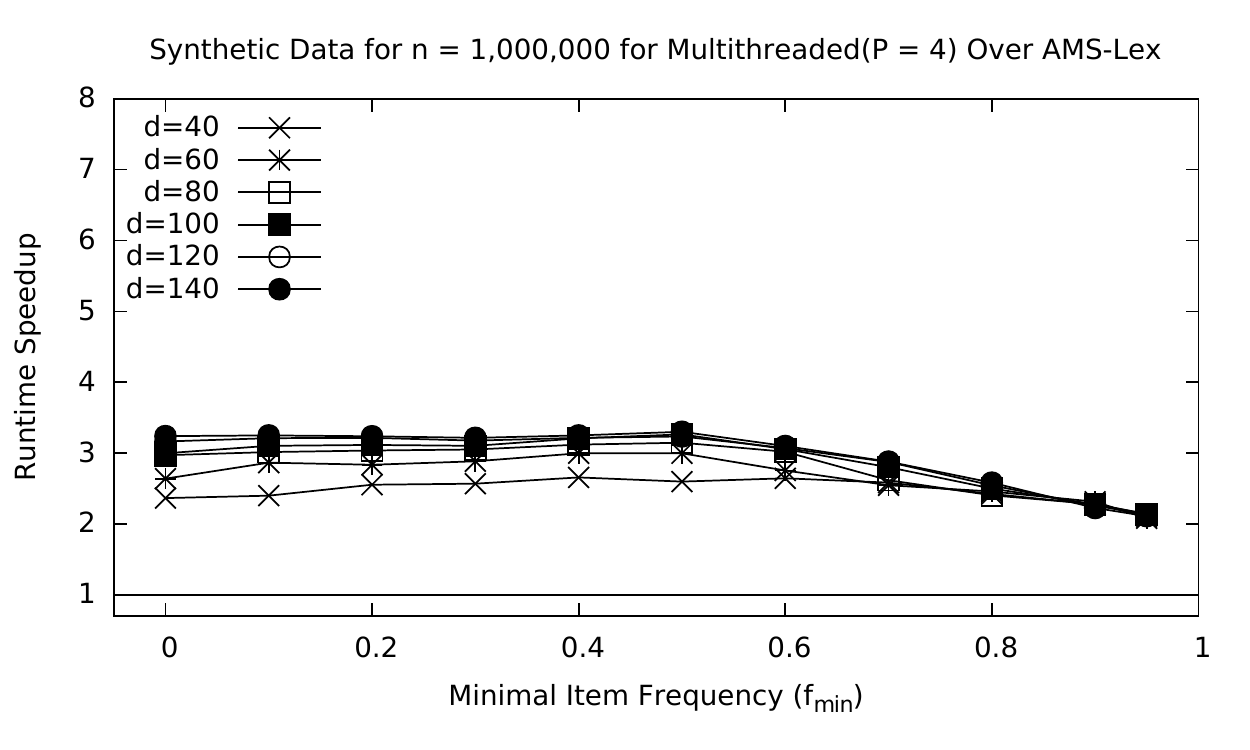}
	\includegraphics[scale=0.5]{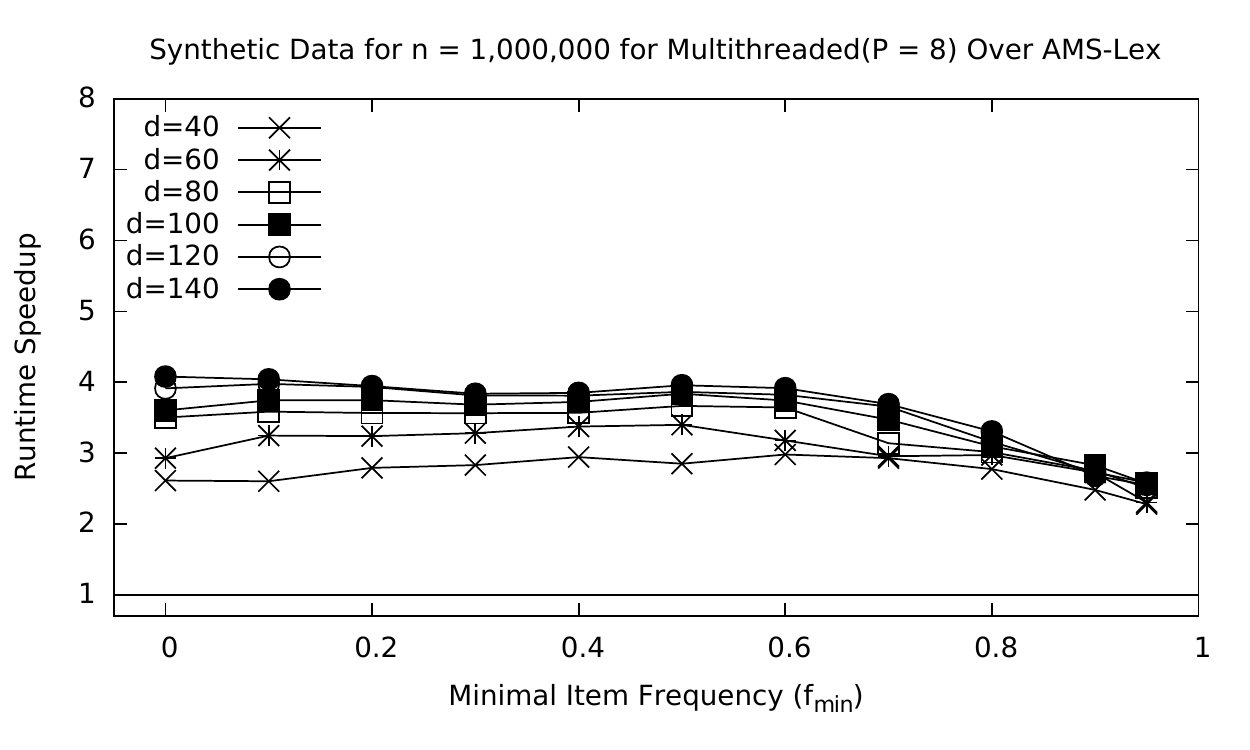}
	\includegraphics[scale=0.5]{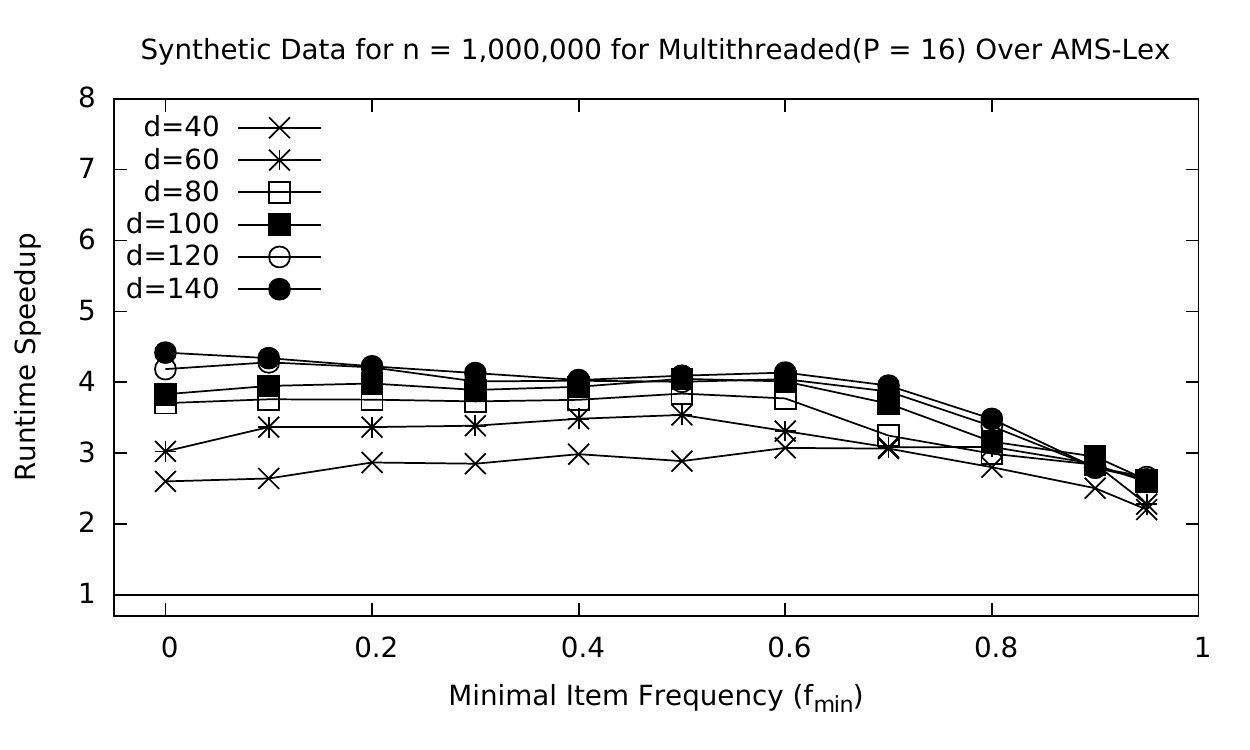}
	\includegraphics[scale=0.5]{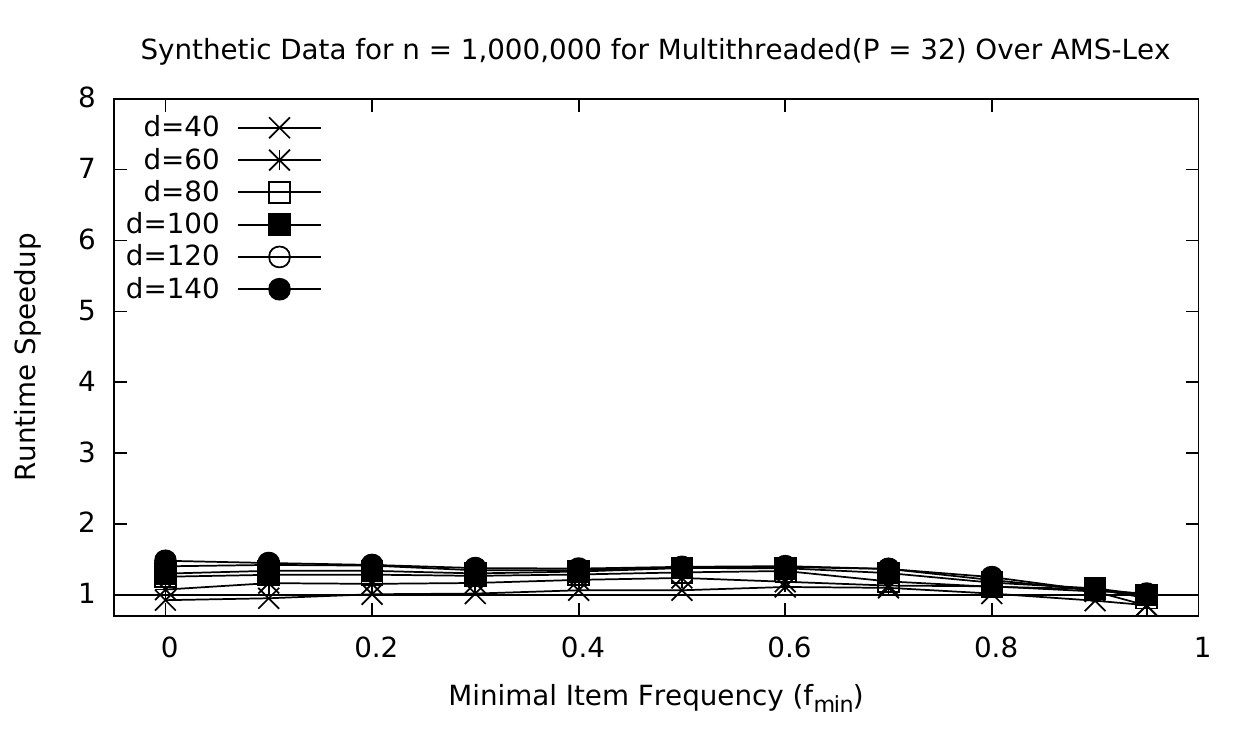}
	\caption{Experimental results for synthetic data for $n = 1\,000\,000$ of comparing our parallel version of AMS-Lex over AMS-Lex for finding the minimal itemsets within a dataset. Here $d$ is the cardinality of the domain of the itemsets. These results show the minimal item frequency described in Section~\ref{sec:experiments} against the resulting execution time speedup of our parallel algorithm compared to AMS-Lex. For these results we have used a machine with $32$ physical cores and used parallelism factors of $4$, $8$, $16$ and $32$ for our parallel modification of AMS-Lex described in section~\ref{sec:parallel}.}
	\label{fig:exp:parallel}
\end{figure}

We have summarised the conducted experiments in Figure~\ref{fig:exp:parallel} which presents the execution time speedup of the parallel algorithm over AMS-Lex using degrees of parallelism $P = 4, 8, 16$ and $32$ on a machine with $32$ physical cores. As input to the algorithm we used datasets with $n = 1\,000\,000$ itemsets with alphabet size of $40$, $60$, $80$, $100$, $120$ and $140$; note that these datasets are the same as the ones used for experimentally comparing the memoized approach versus AMS-Lex consisting of one million itemsets. From the figure, we see that as $d$ increases and keeping $n$ and $f_{min}$ fixed we see that the execution time speedup increases, but it does tend to reach maximum unlike the analogous comparison of memoized over AMS-Lex. We note very small difference in the speedups with $P = 8$ and $P = 16$, whereas as they are both slightly larger then the speedups achieved using $4$ threads.

It is very interesting and important to note that in the case of $P = 32$ we have a significant decay in the speedup over AMS-Lex in comparison to $P = 4, 8$ and $16$. Also, this is the only example we encountered that any of our algorithms is even by a very small amount slower (speedup smaller than $1$ on the graphs) than AMS-Lex. That is explained with the fact that the AMS-Lex algorithm and all of its variations presented here are not computationally intensive but rather memory read access bounded. In this case when $P$ equals the number of physical cores, we found more L3 cache misses in comparison to smaller parallelism factors $P$; also there is a competition for the memory bus and as $P$ increases we inevitably hit the limit of the bus. The cache locality and the memory insensitivity of the application arguments also explains the observed maximum speedups of around $4$ because the machine we used consists of $4$ physical CPU chips, each with its own L3 cache.

\subsubsection{Comparison to Fort et al. GPU Approach}

Fort et al. algorithm for finding extremal sets on a GPU is compared to the AMS-Lex algorithm in \cite{Fort+13}. By carefully analysing the experimental comparison of Fort's algorithm to AMS-Lex, we see that when we exclude the time to pre-process and sort the input dataset to the required format by AMS-Lex then Fort et al. algorithm is between $4$ and $5$ times faster than AMS-Lex when evaluated on synthetic data. Moreover, the execution time speedup demonstrated by the Fort et al. algorithm seems to be constant over AMS-Lex. As presented in Figure~\ref{fig:exp:parallel}, our parallel algorithm is between $3$ and $4.5$ times faster than AMS-Lex when executed with $P = 16$ on a $32$ core machine which is similar to the speedup of Fort et al. algorithm over AMS-Lex. One the other hand, the speedup of our memoized approach over AMS-Lex is not bounded above by a constant as demonstrated. The execution time speedup of our memoized method for datasets with $1\,000\,000$ itemsets over AMS-Lex is as high as $400$ which is much bigger than any speedup reported by Fort et al.~\cite{Fort+13} over AMS-Lex.

\section{Conclusion}
\label{sec:conclusion}

This paper has presented two improved algorithms for identifying extremal sets within a dataset. We have experimentally demonstrated that both techniques improve the performance of the AMS-Lex algorithm on both real world and synthetic datasets. Our first improved algorithm uses memoization to remove redundant work from the AMS-Lex~\cite{BayardoPanda11} requiring at most twice the memory of AMS-Lex. In a subset of the conducted experiments the memoized algorithm executes more than $400$ times faster than AMS-Lex. We show in theory and practice, that the efficiency of this improved algorithm increases as the common prefixes shared by itemsets increases, hence the speedup when compared to AMS-Lex is not bounded above by a constant which is also evident in the experiments provided. The second improved algorithm uses parallelism to speedup the AMS-Lex algorithm. In the conducted experiments we show that our parallel approach outperforms Bayardo and Panda's implementation of AMS-Lex on both real-world and synthetic datasets. Our parallel approach is competitive with Fort~et~al.'s approach running on a highly parallel GPU.

\section{Acknowledgements}
Work supported by the Irish Research Council (IRC) and Science Foundation Ireland grant 12/IA/1381.


\bibliographystyle{elsarticle-num}
\bibliography{Manuscript}

\begin{thebibliography}{10}
\expandafter\ifx\csname url\endcsname\relax
  \def\url#1{\texttt{#1}}\fi
\expandafter\ifx\csname urlprefix\endcsname\relax\def\urlprefix{URL }\fi
\expandafter\ifx\csname href\endcsname\relax
  \def\href#1#2{#2} \def\path#1{#1}\fi

\bibitem{EenBiere05}
N.~E{\'e}n, A.~Biere, Effective preprocessing in sat through variable and
  clause elimination, in: F.~Bacchus, T.~Walsh (Eds.), SAT, Vol. 3569 of
  Lecture Notes in Computer Science, Springer, 2005, pp. 61--75.

\bibitem{Mielikainen+06}
T.~Mielik{\"a}inen, P.~Panov, S.~Dzeroski, Itemset support queries using
  frequent itemsets and their condensed representations, in: L.~Todorovski,
  N.~Lavrac, K.~P. Jantke (Eds.), Discovery Science, Vol. 4265 of Lecture Notes
  in Computer Science, Springer, 2006, pp. 161--172.

\bibitem{BayardoPanda11}
R.~J. Bayardo, B.~Panda, Fast algorithms for finding extremal sets, in: SDM,
  SIAM / Omnipress, 2011, pp. 25--34.

\bibitem{Vieira+09}
M.~R. Vieira, P.~Bakalov, V.~J. Tsotras, On-line discovery of flock patterns in
  spatio-temporal data, in: D.~Agrawal, W.~G. Aref, C.-T. Lu, M.~F. Mokbel,
  P.~Scheuermann, C.~Shahabi, O.~Wolfson (Eds.), GIS, ACM, 2009, pp. 286--295.

\bibitem{Pritchard91}
P.~Pritchard, Opportunistic algorithms for eliminating supersets, Acta Inf.
  28~(8) (1991) 733--754.

\bibitem{BundalaCCSZ14_Optimal_Depth}
D.~Bundala, M.~Codish, L.~Cruz{-}Filipe, P.~Schneider{-}Kamp,
  J.~Z{\'{a}}vodn{\'{y}}, \href{http://arxiv.org/abs/1412.5302}{Optimal-depth
  sorting networks}, CoRR abs/1412.5302.
\newline\urlprefix\url{http://arxiv.org/abs/1412.5302}

\bibitem{Marinov:II:GI-Hard}
M.~{Marinov}, D.~{Gregg}, {Itemset Isomorphism: GI-Hard}, ArXiv e-prints\href
  {http://arxiv.org/abs/1507.05841} {\path{arXiv:1507.05841}}.

\bibitem{Pritchard97}
P.~Pritchard, An old sub-quadratic algorithm for finding extremal sets, Inf.
  Process. Lett. 62~(6) (1997) 329--334.

\bibitem{Yellin92}
D.~M. Yellin, Algorithms for subset testing and finding maximal sets, in: G.~N.
  Frederickson (Ed.), SODA, ACM/SIAM, 1992, pp. 386--392.

\bibitem{YellinJutla93}
D.~M. Yellin, C.~S. Jutla, Finding extremal sets in less than quadratic time,
  Inf. Process. Lett. 48~(1) (1993) 29--34.

\bibitem{Shen96}
H.~Sheni, D.~J. Evans,
  \href{http://www.tandfonline.com/doi/abs/10.1080/00207169608804512}{Fast
  sequential and parallel algorithms for finding extremal sets}, International
  Journal of Computer Mathematics 61~(3-4) (1996) 195--211.
\newblock \href
  {http://arxiv.org/abs/http://www.tandfonline.com/doi/pdf/10.1080/00207169608804512}
  {\path{arXiv:http://www.tandfonline.com/doi/pdf/10.1080/00207169608804512}},
  \href {http://dx.doi.org/10.1080/00207169608804512}
  {\path{doi:10.1080/00207169608804512}}.
\newline\urlprefix\url{http://www.tandfonline.com/doi/abs/10.1080/00207169608804512}

\bibitem{Fort+13}
M.~Fort, J.~A. Sellarès, N.~Valladares, Finding extremal sets on the {GPU},
  Journal of Parallel and Distributed Computing~(0) (2013) --.

\end{thebibliography}

\end{document}